\documentclass[11pt,twoside]{article}

\pdfoutput=1

\usepackage{amsfonts}
\usepackage{amsmath}
\usepackage{amsxtra}
\usepackage{amsthm}
\usepackage{amssymb}
\usepackage{amscd}
\usepackage{eucal}
\usepackage{dsfont}
\usepackage{bm}
\usepackage{graphicx}

\usepackage{a4}
\usepackage{times}
\usepackage{booktabs}
\usepackage{cite}
\usepackage{microtype}
\usepackage[title]{appendix}
\usepackage{float}

\def\beq{\begin{equation}}
\def\eeq{\end{equation}}
\def\bea{\begin{eqnarray}}
\def\eea{\end{eqnarray}}

\def\beann{\begin{eqnarray*}}
\def\eeann{\end{eqnarray*}}

\let\a=\alpha \let\be=\beta \let\g=\gamma \let\de=\delta
\let\e=\varepsilon \let\z=\zeta \let\h=\eta 
\let\eps=\epsilon
 \let\k=\kappa \let\la=\lambda \let\m=\mu
\let\n=\nu \let\x=\xi \let\p=\pi \let\r=\rho \let\s=\sigma
\let\om=\omega \let\ps=\psi
   \let\Ps=\Psi

\let\Om=\Omega  
\let\La=\Lambda \let\G=\Gamma \let\D=\Delta

\let\qd=\quad  

\def\epp{\, .}
\def\epc{\, ,}

\def\tst#1{{\textstyle #1}}

\theoremstyle{plain}
\newtheorem{theorem}{Theorem}
\newtheorem*{theorem*}{Theorem}
\newtheorem{lemma}{Lemma}
\newtheorem*{lemma*}{Lemma}

\newtheorem{corollary}{Corollary}
\newtheorem*{corollary*}{Corollary}

\newtheorem*{conjecture*}{Conjecture}

\theoremstyle{definition}

\newtheorem*{remark}{Remark}
\newtheorem*{question*}{Question}

\def\2{\frac{1}{2}} \def\4{\frac{1}{4}}

\def\6{\partial}

\def\+{\dagger}

\def\<{\langle} \def\>{\rangle}

\def\CO{{\cal O}}

\def\i{{\rm i}}

\def\rd{{\rm d}}

\DeclareMathOperator{\re}{e}

\DeclareMathOperator{\sh}{sh}
\DeclareMathOperator{\ch}{ch}

\DeclareMathOperator{\cth}{cth}

\DeclareMathOperator{\tr}{tr}

\DeclareMathOperator{\supp}{supp}

\DeclareMathOperator{\End}{End}
\DeclareMathOperator{\id}{id}

\DeclareMathOperator{\card}{card}

\DeclareMathOperator{\spin}{{\mathbb S}}
\DeclareMathOperator{\spinflip}{{\mathbb J}}

 \def\Im{{\rm Im\,}}

\def\av{\mathbf{a}}
\def\bv{\mathbf{b}}
\def\cv{\mathbf{c}}

\def\fv{\mathbf{f}}

\def\kv{\mathbf{k}}

\def\tv{\mathbf{t}}
\def\uv{\mathbf{u}}
\def\vv{\mathbf{v}}

\def\xv{\mathbf{x}}
\def\yv{\mathbf{y}}
\def\zv{\mathbf{z}}

\def\fa{\mathfrak{a}}

\newcommand{\lb}{\left}
\newcommand{\rb}{\right}

\allowdisplaybreaks

\begin{document}

\thispagestyle{empty}

\begin{center}

{\Large \bf
Fourth-neighbour two-point functions of the XXZ chain and the
Fermionic basis approach}

\vspace{10mm}

{\large
Frank G\"ohmann,$^\dagger$
Raphael Kleinem\"uhl,$^\dagger$
and Alexander Wei{\ss}e$^\ast$}\\[3.5ex]
$^\dagger$Fakult\"at f\"ur Mathematik und Naturwissenschaften,\\
Bergische Universit\"at Wuppertal,
42097 Wuppertal, Germany\\[1.0ex]
$^\ast$Max-Planck-Institut f\"ur Mathematik, Vivatsgasse 7, 53111 Bonn, Germany

\vspace{40mm}

{\large {\bf Abstract}}

\end{center}

\begin{list}{}{\addtolength{\rightmargin}{9mm}
               \addtolength{\topsep}{-5mm}}
\item
We give a descriptive review of the Fermionic basis approach
to the theory of correlation functions of the XXZ quantum spin
chain. The emphasis is on explicit formulae for short-range
correlation functions which will be presented in a way that
allows for their direct implementation on a computer. Within
the Fermionic basis approach a huge class of stationary
reduced density matrices, compatible with the integrable
structure of the model, assumes a factorized form. This
means that all expectation values of local operators and
all two-point functions, in particular, can be represented
as multivariate polynomials in only two functions $\rho$
and $\omega$ and their derivatives with coefficients that
are rational in the deformation parameter $q$ of the model.
These coefficients are of `algebraic origin'. They do not
depend on the choice of the density matrix, which only impacts
the form of $\rho$ and $\omega$. As an example we work out
in detail the case of the grand canonical ensemble at
temperature $T$ and magnetic field $h$ for $q$ in the critical
regime. We compare our exact results for the fourth-neighbour
two-point functions with asymptotic formulae for $h, T = 0$
and for finite $h$ and $T$.
\end{list}

\clearpage
\section{Introduction}
The study of correlation functions of integrable quantum systems
has remained a challenge for many years now. In our quest for
a general theory we are still proceeding case by case. One
of the most thoroughly studied cases is the spin-$\2$ XXZ quantum
chain. Quantum spin chains are defined on tensor product spaces
${\cal H}_{[k, l]}$, where $k < l \in {\mathbb Z}$, $[k,l] =
\{k, k + 1, \dots, l\}$ and ${\cal H}_{[k,l]} = \bigotimes_{j = k}^l
V_j$. In the spin-$1/2$ case the local vector spaces can be
taken as $V_j = {\mathbb C}^2$. The XXZ-Hamiltonian
\begin{equation} \label{hxxz}
     H_L = J \sum_{j = - L + 1}^L \Bigl\{ \s_{j-1}^x \s_j^x + \s_{j-1}^y \s_j^y
                 + \D \bigl( \s_{j-1}^z \s_j^z - 1 \bigr) \Bigr\}
		 - \frac{h}{2} \sum_{j = - L + 1}^L \s_j^z
\end{equation}
acts on ${\cal H}_{[- L + 1, L]}$. The $\s_j^\a$, $\a = x, y, z$,
are Pauli matrices on $V_j$ and the three real parameters are
the anisotropy $\D = (q + q^{-1})/2$, the exchange interaction
$J > 0$, and the strength $h > 0$ of an external magnetic field.

$H_L$ commutes with the transfer matrix of the 6-vertex model
\cite{McWu68}. Its integrability can be traced back \cite{NirRaz14}
to the quantum group $U_q (\widehat{sl_2})$ \cite{Jimbo85,Drinfeld87,%
ToKh92} using the representation theory of the latter. Finite and
infinite dimensional representations of the quantum group and its
Borel subalgebras are also involved in the `Fermionic basis approach'
\cite{BJMST06b,BJMST08a,JMS08,BJMS09a,JMS13} to the calculation of
the correlation functions of the model which we shall briefly review.
Although this approach was devised already some time ago, we still
have many questions about its interpretation and its connection to
more traditional parts of the theory of integrable quantum systems.
We are not even sure how to properly characterize this fascinating
piece of work in a few words. Perhaps the best we can think of
is to say that it introduces a module structure on a space of
quasi-local spin operators on the infinite chain in a way compatible
with a family of generalized reduced density matrices.

In order to fill these words with meaning and to motivate the basic
notions of the Fermionic basis approach let us take a detour
and recall some general facts about the statistical mechanics of
lattice models. For a many-body system like (\ref{hxxz}) we typically
wish to calculate the correlation functions of local operators in
a macro state described by a density matrix $\r_L$. This is the
state the system relaxes to under the influence of its internal
interactions encoded in its Hamiltonian and of an additional weak
coupling to its environment, if it was initially prepared in
a state given by the experimental setup. Accordingly, $\r_L$ satisfies
the relations
\begin{equation}
     \r_L = \r_L^+ \epc \qd \r_L \ge 0 \epc \qd
     \tr_{[- L + 1, L]} \{\r_L\} = 1 \epc
\end{equation}
valid for any density matrix by definition, and the stationarity
condition
\begin{equation}
     [\r_L, H_L] = 0 \epp
\end{equation}

The fundamental assumption of statistical mechanics is that the
coupling to the environment eventually drives every many-body
system to a state that can be described by the canonical ensemble.
The corresponding canonical density matrix
\begin{equation} \label{candens}
     \r_L^{(c)} (T) = \frac{\re^{- H_L/T}}{\tr_{[- L + 1, L]} \re^{- H_L/T}}
\end{equation}
depends on a single parameter, the temperature $T$. If the interaction
with the environment is very weak, which can be the case e.g.\ in modern
cold-atom experiments, the relaxation to the canonical ensemble
may become very slow and certain transient behaviours may be observed.
For rather long times then the system behaves as if it were isolated.
The question of the relaxation of isolated quantum systems has
therefore become relevant. A natural way of asking this question
is to consider a small subsystem and inquire whether the rest of the
system can act as heat bath for the subsystem. One would say that a
system like (\ref{hxxz}) with density matrix $\r_L (t)$ thermalizes if
\begin{multline} \label{defredd}
     \lim_{t \rightarrow \infty} \lim_{L \rightarrow \infty}
          \tr_{[- L + 1, k - 1] \cup [l + 1, L]} \{\r_L (t)\} \\ =
	  \lim_{L \rightarrow \infty}
          \tr_{[- L + 1, k - 1] \cup [l + 1, L]} \{\r_L^{(c)} (T)\} =
	  D_{[k,l]} (T)
\end{multline}
for all $k, l \in {\mathbb Z}$. The operator on the right hand side
of this equation is called the reduced density matrix (of the
canonical ensemble) associated with the `chain segment'
${\cal H}_{[k,l]}$ or with the `interval' $[k,l]$.

The reduced density matrix defined in (\ref{defredd}) is a very useful
notion as it allows us to define in a sensible way a space of
observables of the infinite chain Hamiltonian associated with (\ref{hxxz}).
The Hamiltonian (\ref{hxxz}) and the canonical density matrix
(\ref{candens}) do not have a naive thermodynamic limit. Such
a limit would require to define a limiting space of states, spanned
by all eigenstates of $H_L$ which have finite excitation energies
for $L \rightarrow \infty$. Since such a construction is not at all
obvious and may depend on the details of the interaction, a better
way to proceed is to define the space of observables inductively
by means of the reduced density matrices of all chain segments.
For this purpose consider any $X_{[k,l]} \in \End {\cal H}_{[k,l]}$
and let $X_L = \id_{[- L + 1, k - 1]} \otimes X_{[k,l]} \otimes
\id_{[l + 1, L]}$ for $L$ large enough, such that $[k,l] \subset
[- L + 1, L]$. Then 
\begin{equation} \label{defexpxinf}
     \<X\>_T = \lim_{L \rightarrow \infty}
               \tr_{[- L + 1, L]} \{\r_L^{(c)} (T) X_L\}
         = \tr_{[k,l]} \{D_{[k,l]} (T) X_{[k,l]}\} \epp
\end{equation}
This equation may be interpreted as defining the action of
$X_{[k,l]}$ on an infinite chain by formally setting
\begin{equation} \label{xinf}
     X = \id_{(- \infty, k - 1]} \otimes X_{[k,l]} \otimes \id_{[l + 1, \infty)} \epp
\end{equation}
Every operator on the infinite chain, which can be represented
like this, will be called a local operator. Let $[k,l]$ be
the minimal interval for which $X$ has a representation like~%
(\ref{xinf}). Then $X_{[k,l]}$ is called the non-trivial part
of $X$, $[k, l]$ (or ${\cal H}_{[k,l]}$) is called its support,
$\supp (X)$, and $\ell (X) = \card [k,l]$ its length. $X = \id$
is the unique operator of length zero. Clearly the local operators
on the infinite chain span a vector space~$\cal W$.

If the system is initially represented by an ensemble with density matrix
$\r_L$ that thermalizes in the sense of (\ref{defredd}), the expectation
value of every local operator $X$ evolves in time to its canonical
expectation value determined by a reduced density matrix,
\begin{multline}
     \lim_{t \rightarrow \infty} \<X(t)\> =
        \lim_{t \rightarrow \infty} \lim_{L \rightarrow \infty}
	\tr_{[- L + 1, L]} \{\r_L X_L (t)\} =
        \lim_{t \rightarrow \infty} \lim_{L \rightarrow \infty}
	\tr_{[- L + 1, L]} \{\r_L (t) X_L\} \\
            = \tr_{[k,l]} \{D_{[k,l]} (T) X_{[k,l]}\} \epp
\end{multline}
We are not aware of a proof of this statement, not even for
quantum spin systems on 1d lattices, but this behaviour is quite
universally observed in experiments. One should have in mind
however, that different equilibrium ensembles are equivalent
if they produce the same reduced density matrices for $L
\rightarrow \infty$.

The above mentioned cold-atom experiments suggest that thermalization 
will not happen with the Heisenberg time evolution $X_L \mapsto
\re^{\i H_L t} X_L \re^{- \i H_L t}$ if $H_L$ is the Hamiltonian
of an integrable quantum chain such as (\ref{hxxz}). In this case
a coupling to a bath, $H_L \rightarrow H_L + H_{\rm bath}$, is
required for thermalization. The relaxation of integrable lattice
systems has been a subject of intensive debate over the past decade.
For a review covering most of the above discussion and extending
it in several directions see \cite{EsFa16}. Some of the natural
questions that arise in this context are: Do isolated integrable
systems relax at all? What is the space of all reduced density matrices
a given integrable system can relax to? What are reasonable classes
of initial density matrices of many-body systems that can be
realized in experiments? Experimentalists have provided an answer
to the latter question. They can realize so-called quenches. For
these the initial macro state is assumed to be represented by the
projector onto the ground state sector of the Hamiltonian for
a certain set of interaction parameters (like a certain value
of $\D$ in (\ref{hxxz})) that are then suddenly changed at initial
time $t = 0$. The other questions are still not fully answered, at
least not with sufficient rigour. It is believed that integrable
many-body systems do relax even if they are isolated \cite{ViRi16,EsFa16}.
There is also a certain amount of evidence that the reduced density
matrices that describe the system after relaxation are related to
certain generalized Gibbs ensembles \cite{RDYO07,ViRi16,EsFa16}.

Consider a system with $N$ local conserved charges $\{H^{(n)}_L\}_{n=1}^N$
that mutually commute among each other, one of which, $H^{(1)}_L$ say,
is its Hamiltonian. Then a natural generalization of the canonical
density matrix (\ref{candens}) is
\begin{equation} \label{gencandens}
     \r_L^{(N)} (\be_1, \dots, \be_N) =
        \frac{\re^{- \sum_{n=1}^N \be_n H_L^{(n)}}}
	     {\tr_{[- L + 1, L]} \{\re^{- \sum_{n=1}^N \be_n H_L^{(n)}}\}}
\end{equation}
which satisfies a maximum entropy condition under the constraint
that the ensemble averages of the conserved charges $H_L^{(n)}$
are fixed. The density matrices $\r_L^{(N)}$ generate a sequence
of reduced density matrices
\begin{equation} \label{defgenredd}
     D_{[k,l]}^{(N)} (\be_1, \dots, \be_N) =
        \lim_{L \rightarrow \infty}
        \tr_{[- L + 1, k - 1] \cup [l + 1, L]} \{\r_L^{(N)} (\be_1, \dots, \be_N)\}
\end{equation}
which describe the infinite system in formally the same way as
$D_{[k,l]} (T)$ in the canonical case.

We shall call a quantum spin chain integrable, if its Hamiltonian
$H_L$ commutes with a commuting family of transfer matrices $t_L (\z)$
with spectral parameter $\z$, whose local commutativity condition
is the Yang-Baxter equation. Typically, for $L \rightarrow \infty$,
integrable systems have infinitely many local conserved charges generated
by the function $\ln t_L^{-1} (0) t_L (\z)$. For this reason $N$ in
(\ref{defgenredd}) is often formally sent to infinity in the physics
literature. This may be interpreted in the following way. Suppose
$\bigl( (\be_{n, N})_{n=1}^N\bigr)_{N \in {\mathbb N}}$ is a sequence
of real numbers such that the limits $\lim_{N \rightarrow \infty}
\be_{n, N} = \be_n$ and $\lim_{N \rightarrow \infty}
D_{[k,l]}^{(N)} (\be_{1, N}, \dots, \be_{N, N}) = D_{[k,l]}
\bigl((\be_n)_{n \in {\mathbb N}}\bigr)$ exist for every fixed
interval $[k,l]$. This would then define a sequence of reduced
density matrices in much the same way as in the canonical case.
Still, the existence of such a limit, a clear description of the
space of admissible sequences $(\be_n)_{n \in {\mathbb N}}$ and
the `completeness of the set of local operators' \cite{IDWCEP15}
are questions that will be hard to answer in full generality.

On the other hand, there is a huge class of reduced density
matrices that appeared in studies of the correlation functions
of integrable lattice models and is compatible with their integrable
structure. We assume that the reader is familiar with the
graphical representation of vertex models (otherwise please see
e.g.\ \cite{Goehmann20}). The generalized reduced density matrices
we are referring to can be graphically represented as
\begin{equation} \label{defdensmatapp}
     D_{[k,l]}^{(N)} (\k) \: = \:
        \text{\raisebox{-64pt}{\includegraphics[width=.54\textwidth]{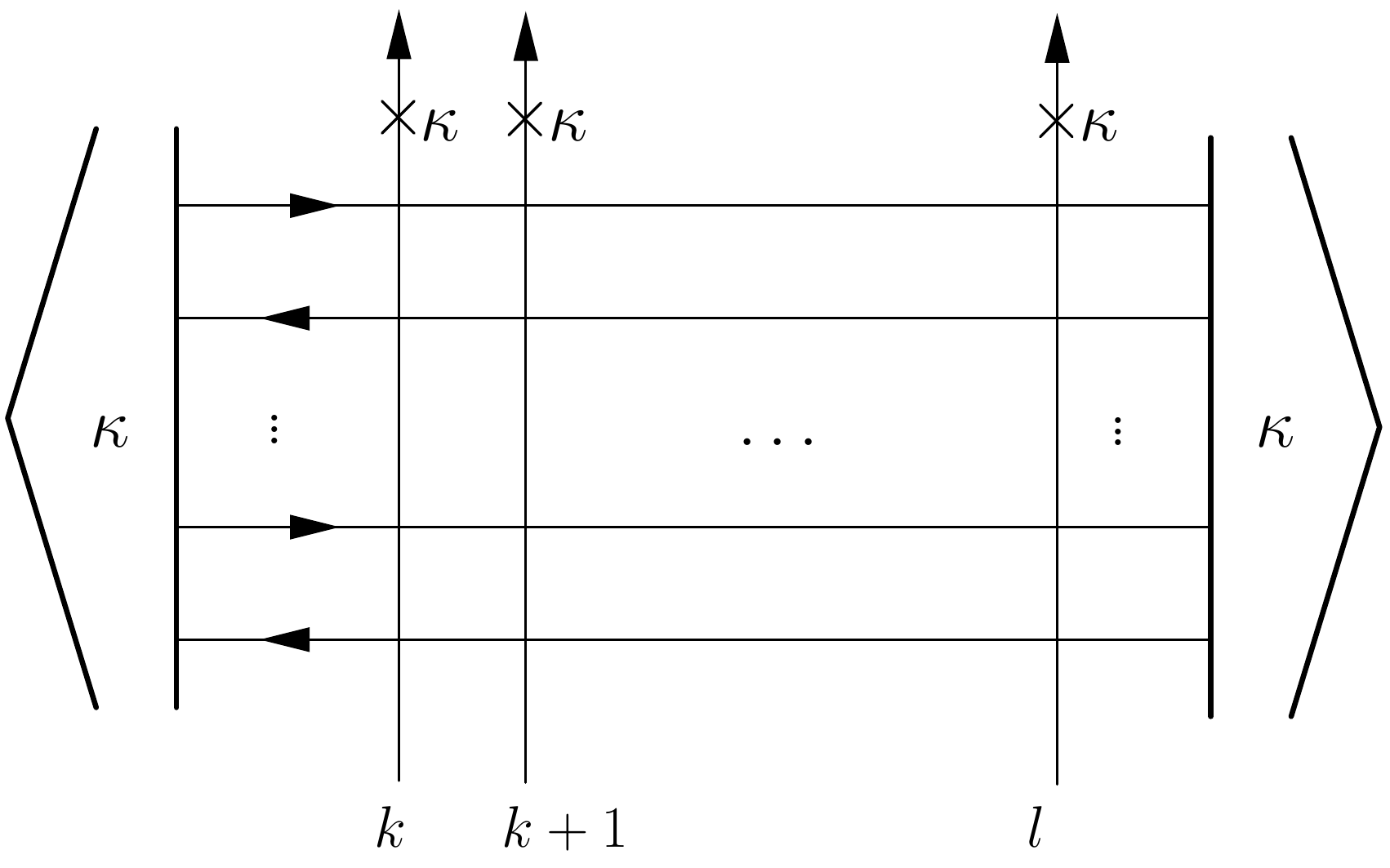}}}
	\times \frac{\La^{k - l - 1} (0|\k)}{\<\k|\k\>} \epp
\end{equation}
Here every line crossing represents a $U(1)$-symmetric $R$-matrix,
the crosses stand for the corresponding local `gauge fields' (or
twists), and $|\k\>$ is an eigenstate of the column-to-column transfer
matrix with eigenvalue $\La(\z|\k)$. The index `$N$' refers to $N$ horizontal
spectral parameters $(\m_{n, N})_{n=1}^N$. The construction of the
generalized reduced density matrix (\ref{defdensmatapp}) offers
considerable freedom as the horizontal spectral parameters as well as their
number are arbitrary. This freedom can be used to realize the reduced
density matrix of the canonical ensemble \cite{GKS04a,GKS05} or, with
little more effort \cite{KlSa02}, reduced density matrices of the form
(\ref{gencandens}) (see \cite{FaEs13,Pozsgay13} for an application to
quenches in the XXZ chain). Another important class of reduced density
matrices which can be represented by (\ref{defdensmatapp}) are those
connected with projectors onto eigenstates $|n\>$ of $t_L (\z)$, i.e.\
reduced density matrices of the form $D_{[k,l]} (L) = \tr_{[- L + 1, k - 1]
\cup [l + 1, L]} \{|n\>\<n|\}$.

At least for the XXZ chain, considered from now on in this work, the
algebraic Bethe Ansatz \cite{STF79} and Slavnov's scalar product formula
\cite{Slavnov89} can be used in order to rewrite (\ref{defdensmatapp})
as an $(l - k + 1)$-fold integral. This technique was pioneered
for the ground state of the infinite chain in \cite{KMT99b}. In
\cite{GKS04a,GKS05} it was shown that similar techniques when
combined with the methodology of non-linear integral equations
\cite{Kluemper93} work in the finite temperature case as well. An
exploration of the case of excited states at finite length was provided
in \cite{Pozsgay17}. In \cite{BoKo01,BGKS06} it was shown that the
multiple integrals factorize which makes it possible, in principle,
to evaluate the short-distance correlation functions of the model.
The deeper reason for the factorization was unraveled in a remarkable
series of works \cite{BJMST04a,BJMST04b,BJMST05b,BJMST06b,BJMST08a}
that culminated in the discovery of the Fermionic basis and in
the proof of what we like to call the JMS theorem \cite{JMS08}.

Note that, in general, equation (\ref{defdensmatapp}) does not
define a density matrix, since the expression on the right hand
side is not necessarily Hermitian or positive. This is why
we call $D_{[k,l]}^{(N)} (\k)$ a generalized reduced density
matrix. On the other hand, it still satisfies the normalization
condition $\tr_{[k,l]} \bigl\{D_{[k,l]}^{(N)} (\k)\bigr\} = 1$ and, more
generally, the reduction relations
\begin{equation}
     \tr_k \bigl\{D_{[k,l]}^{(N)} (\k) \bigr\}= D_{[k + 1,l]}^{(N)} (\k) \epc \qd
     \tr_l \bigl\{D_{[k,l]}^{(N)} (\k) \bigr\}= D_{[k,l - 1]}^{(N)} (\k) \epp
\end{equation}
This is enough to define generalized expectation values of local operators
on the infinite chain. An interesting question would be to describe
the space of all reduced density matrices of the form (\ref{defdensmatapp}),
which would mean to characterize all admissible sets of horizontal
spectral parameters that render $D_{[k,l]}^{(N)} (\k)$ Hermitian and
positive.

We have included this somewhat lengthy reflection on the equilibration
of integrable quantum systems and on reduced density matrices in order
to motivate some of the central notions of the Fermionic basis approach
to the correlation functions of the XXZ chain that will hopefully
appear rather natural now. These are the space of quasi-local operators
and a further generalization of the reduced density matrices. Following
\cite{BJMST08a} we define the spin operator
\begin{equation} \label{szkl}
     S_{[k,l]} = \2 \sum_{j=k}^l \s_j^z \in \End {\cal H}_{[k,l]}
\end{equation}
and its formal extension to the infinite chain 
\begin{equation}
     S(k) = S_{(- \infty, k]} \otimes \id_{[k+1, \infty)} \epp
\end{equation}
This allows us to deform the notion of local operators introduced
in (\ref{xinf}). We fix $\a \in {\mathbb C}^\times$. For any local
operator $X$ with non-trivial part $X_{[k,l]}$ we set
\begin{equation} \label{alphaxinf}
     X^{(\a)} = q^{2 \a S(k-1)} X \epp
\end{equation}
We call $X^{(\a)}$ a quasi-local operator with tail $q^{2 \a S(k-1)}$.
The quasi-local operators span a vector space ${\cal W}_\a$. The
concepts of the non-trivial part and of the length of operators carry
over from ${\cal W}$ to ${\cal W}_\a$. The generalized density matrix
(\ref{defdensmatapp}) can be deformed in a way that is compatible
with the $\a$-deformation of the vector space,
\begin{equation} \label{gendensmat}
     D_{[k,l]}^{(N)} (\k, \a) \: = \:
        \text{\raisebox{-64pt}{\includegraphics[width=.54\textwidth]{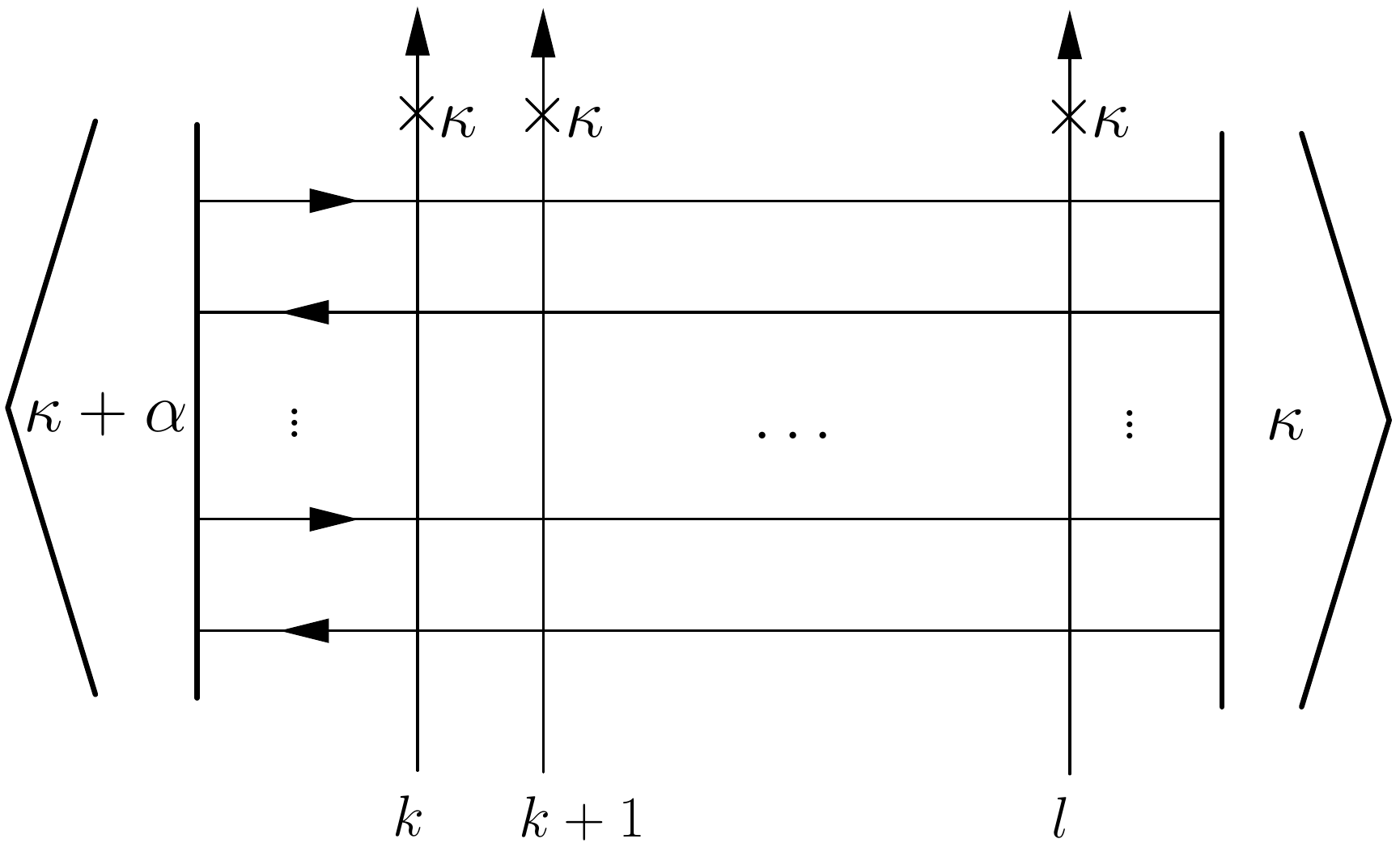}}}
	\times \frac{\La^{k - l - 1} (0|\k)}{\<\k + \a|\k\>} \epp
\end{equation}
Here arbitrary finite-dimensional representations on the horizontal
lines will be admitted \cite{JMS08}. The most important generalization
as compared to (\ref{defdensmatapp}) is, however, that we consider
eigenstates $|\k\>$, $|\k + \a\>$ of twisted transfer matrices with
different twists, $\k$ and $\k + \a$, to the left and to the right.
Interestingly, this object can still be represented (see appendix of
\cite{BoGo09}) as a multiple integral of the same form as for the
reduced density matrix pertaining to the ground state \cite{JMMN92,%
JiMi96,KMT99b} or to the canonical ensemble at finite temperature
\cite{GKS05}.

$ D_{[k,l]}^{(N)} (\k, \a)$ can be used \cite{JMS08} to define
generalized ensemble averages $Z^\k: {\cal W}_\a \rightarrow
{\mathbb C}$,
\begin{multline} \label{defzkappa}
     Z^\k [X^{(\a)}] =
        \lim_{L \rightarrow \infty}
	\frac{\tr_{[- L + 1, L]} \bigl\{D_{[- L + 1, L]}^{(N)} (\k, \a) \,
              q^{2 \a S_{[- L + 1, k - 1]}} \otimes X_{[k,l]}
	      \otimes \id_{[l + 1, L]} \bigr\}}
             {\tr_{[- L + 1, L]} \{D_{[- L + 1, L]}^{(N)} (\k, \a) \,
	      q^{2 \a S_{[- L + 1, k - 1]}} \otimes \id_{[k, L]}\}} \\[1ex]
     = \tr_{[k, l]} \bigl\{D_{[k, l]}^{(N)} (\k, \a) \, X_{[k,l]} \bigr\} \epp
\end{multline}
This clearly generalizes (\ref{defexpxinf}). Our above discussion
shows that $Z^\k$ can be employed to realize the canonical ensemble and
ground state averages, averages with respect to excited states of
finite chains and averages with respect to generalized Gibbs ensembles
including any number of local conserved charges. It is tempting to
speculate that  all ensembles an isolated XXZ chain can relax to have
reduced density matrices that can be represented by means of $Z^\k$ (in
general a limit $N \rightarrow \infty$ will have to be considered).

The spin provides a natural grading of the space ${\cal W}_\a$. For the
adjoint action of the spin operator (\ref{szkl}) we shall write
${\mathbb S}_{[k,l]} = [S_{[k,l]}, \cdot \:]$. By definition
$X_{[k,l]} \in \End {\cal H}_{[k,l]}$ has spin $s \in {\mathbb Z}$ if
${\mathbb S}_{[k,l]} X_{[k,l]} = s X_{[k,l]}$. This carries over to
the space ${\cal W}_\a$. Let $X^{(\a)} \in {\cal W}_\a$ with non-trivial
part $X_{[k,l]}$. Define ${\mathbb S}: {\cal W}_\a \rightarrow
{\cal W}_\a$ by ${\mathbb S} X^{(\a)} = [S(\infty), X^{(\a)}]$.
$X^{(\a)} \in {\cal W}_\a$ has spin $s$ if ${\mathbb S} X^{(\a)}
= s X^{(\a)}$. The spin of $X^{(\a)}$ is equal to the spin of its
non-trivial part. Let ${\cal W}_{\a, s} \subset {\cal W}_\a$ be the
subspace of all operators with spin~$s$. Clearly ${\cal W}_\a
= \bigoplus_{s \in {\mathbb Z}} {\cal W}_{\a, s}$ and
$Z^\k [{\cal W}_{\a, s}] = 0$ for all $s \ne 0$. Thus, only the
subspace ${\cal W}_{\a, 0}$ is relevant for the calculation of
correlation functions.

As we shall see, the operators that generate the Fermionic basis
change the tail and the spin of a quasi-local operator. They can
be constructed as acting on the space
\begin{equation} \label{defwalpha}
     {\cal W}^{(\a)} = \bigoplus_{s \in {\mathbb Z}} {\cal W}_{\a - s, s}
\end{equation}
rather than on ${\cal W}_\a$. Note that ${\cal W}_{\a, 0} = {\cal W}_\a
\cap {\cal W}^{(\a)}$. Iterating the above definition we say that an
operator $\xv \in \End {\cal W}^{(\a)}$ has spin $s$ if $[{\mathbb S}, \xv]
= s \xv$, and similarly in the finite case, where $\xv_{[k,l]} \in
\End \End {\cal H}_{[k,l]}$ has spin $s$, if $[{\mathbb S}_{[k,l]},
\xv_{[k,l]}] = s \xv_{[k,l]}$.

In \cite{BJMST06b,BJMST08a} the authors have constructed a
module structure on ${\cal W}^{(\a)}$ that is generated by the
coefficients of the following formal series,
\begin{subequations}
\label{modes}
\begin{align}
     \tv^* (\z) & = \sum_{p=1}^\infty (\z^2 - 1)^{p-1} \tv^*_p \epc \\
     \xv^* (\z) & = \z^{\a s(\xv^*) + 2} \sum_{p=1}^\infty (\z^2 - 1)^{p-1} \xv^*_p \epc \\
     \xv (\z) & = \z^{\a s(\xv)} \sum_{p=0}^\infty \frac{\xv_p}{(\z^2 - 1)^p} \epc
     \label{annihimodes}
\end{align}
\end{subequations}
where $\xv^* = \bv^*, \cv^*$ and $\xv = \bv, \cv$. The coefficients
$\tv^*_p, \xv^*_p, \xv_p \in \End {\cal W}^{(\a)}$ will be called
the modes and the formal series the `mode expansions' of the Fermionic
operators. The modes have definite spins
\begin{equation}
     s(\cv_p^*) = s(\bv_p) = - 1 \epc \qd
     s(\tv_p^*) = 0 \epc \qd
     s(\bv_p^*) = s(\cv_p) = 1 \epp
\end{equation}
Moreover, they exhibit the following block structure,
\begin{equation}
     \xv_p : {\cal W}_{\a - s, s} \rightarrow
        {\cal W}_{\a - s - s(\xv_p), s + s(\xv_p)} \epp
\end{equation}
This is what we meant when we said that the operators that generate
the Fermionic basis change the tail and the spin of a quasi-local
operator. The effect of the modes on the length of operators is also
well controlled. For any $X^{(\a)} \in {\cal W}^{(\a)}$ we have the
inequalities
\begin{align}
     \ell \bigl(X^{(\a)}\bigr) \le
     \ell \bigl(\xv_p^* \, X^{(\a)}\bigr) \le  \ell \bigl(X^{(\a)}\bigr) + p \epc \qd
     \ell \bigl(\xv_p \, X^{(\a)}\bigr) \le  \ell \bigl(X^{(\a)}\bigr) \epc
\end{align}
where $\xv^* = \tv^*, \bv^*, \cv^*$ and $\xv = \bv, \cv$. Moreover,
\begin{equation} \label{annihipropmodes}
     \xv_p \, X^{(\a)} = 0
\end{equation}
for $\xv = \bv, \cv$ if $\ell \bigl(X^{(\a)}\bigr) < p$, whence the name
annihilation operators for $\bv$ and $\cv$. The operators  $\xv^* = \tv^*,
\bv^*, \cv^*$ will be called creation operators. They owe their name to 
the following
\begin{theorem} \label{thm:basis}
Boos et al.\ \cite{BJMS09a}. The set
\begin{equation} \label{deffb}
     {\cal F} =
        \biggl\{{\bm \tau}^m \tv_{p_1}^* \dots \tv_{p_j}^*
                \bv_{q_1^+}^* \dots \bv_{q_k^+}^* \cv_{q_1^-}^* \dots \cv_{q_k^-}^* \,
		q^{2 \a S(0)} \bigg|
		\begin{array}{l}
		{\scriptstyle m \in {\mathbb Z};\: j, k \in {\mathbb Z}_{\ge 0};} \\[-.5ex]
		{\scriptstyle p_1 \ge \dots \ge p_j \ge 2}, \:
		{\scriptstyle q_1^\pm > \dots > q_k^\pm \ge 1}
		\end{array}
		\biggr\}
\end{equation}
is a basis of ${\cal W}_{\a, 0}$. Here ${\bm \tau} = \tv_1^*/2$ is the
right shift operator.
\end{theorem}
The interpretation of this theorem is that ${\cal W}_{\a, 0}$ is
a module generated by the action of the creation operators on
a Fock vacuum $q^{2 \a S(0)}$. The creation operators $\tv_p^*$
are central. They commute among themselves and with all other
creation and annihilation operators. The modes $\bv_p^*$, $\bv_q$
and $\cv_p^*$, $\cv_q$ are two sets of creation and annihilation
operators of Fermions. Their mutual anticommutators all vanish
except for
\begin{equation}
     [\bv_p, \bv_q^*]_+ = \de_{p,q} \epc \qd [\cv_p, \cv_q^*]_+ = \de_{p,q} \epp
\end{equation}
For this reason $\cal F$ is called the Fermionic basis.

The commutation and anticommutation rules of the modes easily
follow from similar relations among the Fermionic operators.
\begin{equation} \label{commutism}
     [\tv^* (\z_1), \xv (\z_2)] = 0 \epc \qd \xv = \tv^*, \bv^*, \cv^*, \bv, \cv \epc
\end{equation}
while all anticommutators among $\bv^* (\z)$, $\cv^* (\z)$, $\bv (\z)$
and $\cv (\z)$ are vanishing, except for
\begin{equation} \label{anticommutism}
     [\bv (\z_1), \bv^* (\z_2)]_+ = - \ps (\z_2/\z_1, \a) \epc \qd
     [\cv (\z_1), \cv^* (\z_2)]_+ = \ps (\z_1/\z_2, \a) \epc
\end{equation} 
where
\begin{equation}
     \ps(\z, \a) = \frac{\z^\a (\z^2 + 1)}{2(\z^2 - 1)} \epp
\end{equation}
The proof of the commutation and anticommutation relations
for the mode expansions is probably the most involved part 
of the construction of the Fermionic basis. It can be found
in \cite{BJMST08a} and \cite{JMS13}.

The construction of the creation and annihilation operators in
\cite{BJMST08a} is rather explicit. These operators are constructed
from finite building blocks whose action is inductively extended to the
infinite chain. This is possible due to certain `reduction properties'
of the finite building blocks. We shall review the construction
of these finite building blocks and the reduction properties in the
following section. We shall write them in a form that is suitable
for their construction by means of computer algebra programs.
This will allow us to obtain explicit expressions for the
fourth-neighbour two-point functions.

The most remarkable property of the creation and annihilation
operators is their compatibility with the functional $Z^\k$
described above. It reveals itself in the following fundamental
theorem due to Jimbo, Miwa and Smirnov.
\begin{theorem} \label{theo:jms}
JMS theorem \cite{JMS08}. For any given realization of the functional
$Z^\k$ defined in (\ref{defzkappa}) there are two functions
$\r (\z)$ and $\om (\z, \x; \a)$ such that, for every $X^{(\a)} \in
{\cal W}_\a$,
\begin{subequations}
\begin{align}
     Z^\k \bigl\{\tv^* (\z) X^{(\a)}\bigr\} &
        = 2 \r (\z)  Z^\k \bigl\{X^{(\a)}\bigr\} \epc \\[1ex]
     Z^\k \bigl\{\bv^* (\z) X^{(\a + 1)}\bigr\} &
        = \int_\G \frac{\rd \x^2}{2 \p \i \x^2} \:
	     \om (\z, \x; \a) Z^\k \bigl\{\cv (\x) X^{(\a + 1)}\bigr\} \epc \\[1ex]
     Z^\k \bigl\{\cv^* (\z) X^{(\a - 1)}\bigr\} &
        = - \int_\G \frac{\rd \x^2}{2 \p \i \x^2} \:
	     \om (\x, \z; \a) Z^\k \bigl\{\bv (\x) X^{(\a - 1)}\bigr\}
\end{align}
where $\G$ is a simple closed contour around $\x^2 = 1$.
\end{subequations}
\end{theorem}
The functions $\r$ and $\om$ are described in \cite{JMS08}. While
$\r$ is simply the ratio of the eigenvalues
\begin{equation} \label{defrho}
      \r (\z) = \frac{\La(\z|\k + \a)}{\La(\z|\k)}
\end{equation}
pertaining to the left and right eigenvectors $\<\k + \a|$ and
$|\k\>$, the definition of $\om(\z, \x; \a)$ in the general case
is rather involved. We do not reproduce it here, but refer to
\cite{JMS08} for the details. A characterization of $\om$ in
terms of solutions of nonlinear integral equations for the two
important special cases of the canonical density matrix and
of the reduced density matrix connected with the ground state
of a finite system of length $2L$ was obtained in \cite{BoGo09}.
Such kind of characterization turned out to be useful in
applications of the Fermionic basis approach to the description
of correlation functions of quantum field theories \cite{JMS11,BJMS10}.
We shall recall the finite temperature case below when we consider
the finite temperature short-range two-point functions.

The importance of the JMS theorem for the calculation of correlation
functions reveals itself in the following two corollaries to
Theorem~\ref{theo:jms}.
\begin{corollary} \label{cor:zkappagenfns}
Theorem~\ref{theo:jms} together with the commutation relations
of the creation and annihilation operators implies that
\begin{multline}
     Z^\k \bigl\{\tv^* (\z_1^0) \cdots \tv^* (\z_k^0)
                 \bv^* (\z_1^+) \cdots \bv^* (\z_l^+) \cv^* (\z_l^-)
		 \cdots \cv^* (\z_1^-) \, q^{2 \a S(0)} \bigr\} \\[1ex]
        = \biggl[\prod_{i=1}^k 2 \r (\z_i^0)\biggr]
	  \det_{m, n = 1, \dots, l} \bigl\{\om (\z_m^+, \z_n^-; \a)\bigr\} \epp
\end{multline}
\end{corollary}
Using this corollary and the mode expansions (\ref{modes}) we can
evaluate $Z^\k$ on the Fermionic basis (\ref{deffb}).
\begin{corollary} \label{cor:baseexp}
\begin{multline} \label{expfbasis}
     Z^\k \bigl\{\tv^*_{p_1} \cdots \tv^*_{p_k}
                 \bv^*_{q_1^+} \cdots \bv^*_{q_l^+} \cv^*_{q_l^-}
		 \cdots \cv^*_{q_1^-} \, q^{2 \a S(0)} \bigr\}
        = \Biggl[\prod_{i=1}^k
	         \frac{\6_{(\z_i^0)^2}^{p_i - 1}
		       2 \r (\z_i^0)\bigr|_{(\z_i^0)^2 = 1}}{(p_i - 1)!}\Biggr] \\[1ex]
	  \times \det_{m, n = 1, \dots, l} \Biggl\{
	         \frac{\6_{(\z_m^+)^2}^{q_m^+ - 1}}
		      {(q_m^+ - 1)!}
	         \frac{\6_{(\z_n^-)^2}^{q_n^- - 1}}
		      {(q_n^- - 1)!} \biggl(\frac{\z_n^-}{\z_m^+}\biggr)^\a
		      \frac{\om (\z_m^+, \z_n^-; \a)}{(\z_m^+)^2 (\z_n^-)^2}
		      \biggr|_{(\z_m^+)^2 = (\z_n^-)^2 = 1} \Biggr\} \epp
\end{multline}
\end{corollary}
Since ${\cal F}$ is a basis of ${\cal W}_{\a, 0}$ it follows that
the generalized `$Z^\k$ expectation value' of every quasi-local
operator is a polynomial in $\r (\z)$, $\om(\z, \x; \a)$ and the
derivatives of these functions with respect to their arguments $\z$,
$\x$ at $\z = \x = 1$. This is true, in particular, for the finite
temperature and finite length ground state expectation values which
require to perform the limit $\a \rightarrow 0$ in the end. We call
this property of the physical expectation values `factorization'.
The Fermionic basis approach and the above corollary can be
understood as providing the algebraic reason for the factorization
of multiple integrals observed in \cite{BoKo01,BGKS06}.

The Fermionic basis $\cal F$ has a number of specific properties
worked out in \cite{BJMS09a}. First of all, the action of the creation
operators extends the support of a given operator $X$ only to the
right, not to the left. This is why negative powers of $\bm \tau$ are
needed in (\ref{deffb}). Consider the subspace ${\cal W}_{\a, 0, [1,n]}
\subset {\cal W}_{\a, 0}$ of quasi-local operators $X^{(\a)}$ with
$\supp (X^{(\a)}) \subset [1,n]$. In \cite{BJMS09a} the authors
provide an explicit construction of a Fermionic basis of this $4^n$-%
dimensional subspace of ${\cal W}_{\a, 0}$. They show that it
is generated by the action of polynomials in the modes in which
every term has at most $n$ factors and contains no modes higher than
the $n$th. A basis of ${\cal W}_{\a, 0, [1,1]}$, for instance, is
generated by the action of $\bv_1^*$, $\cv_1^*$, $\tv_1^*$ and $\id$
on the Fock vacuum $q^{2 \a S(0)}$. In general, the action of finite
products of modes can be represented by finite matrices. This makes
it possible to calculate short-range correlation functions. The
non-trivial part of the image of the Fock vacuum under a product
of modes of creation operators can be represented by finite
matrices which, in turn, can be written as products of operators
each acting on a single `lattice site' $V_j$. These `ultra-local'
operators can be represented using any basis of $\End V_j$, for
instance the Pauli matrices together with the $2 \times 2$ unit
matrix.

Unfortunately, no closed formula for the `inverse problem' of
expanding products of ultra-local operators in the Fermionic
basis is known. This makes the direct use of Corollary~\ref{cor:baseexp}
for the calculation of short-range correlation functions rather
inefficient. By direct calculation we were only able to proceed
up to operators of lengths three \cite{Kleinemuehl20}. A more efficient
algorithm for the calculation of the coefficients of the basis
transformation from the Fermionic basis to the standard basis of
ultra-local operators was devised in \cite{FrSm18}. The basic idea
of this work is to exploit the freedom in the definition of the
functional~$Z^\k$. The authors consider certain particularly simple
realizations of $Z^\k$ for which they use a clever algorithm in
order to directly calculate the $Z^\k$-expectation values of the
operators under consideration. For the same realizations they can
independently calculate the function $\om$ which, according to
(\ref{expfbasis}), determines the expectation values of the elements
of the Fermionic basis. Proceeding like this for many different 
realizations of $Z^\k$ they obtain a linear (and typically
overdetermined) system of equations for the expansion coefficients
which can be solved on a computer. In \cite{FrSm18, MiSM19} this
is worked out for the XXX chain in zero external field. In
\cite{Smirnov18pp} the case of quantum group invariant operators
on the XXZ chain is considered.

In this work we are going to explore another possibility to
calculate short-range correlation functions by means of the
Fermionic basis approach. We will be using the so-called
exponential form of the density matrix introduced in \cite{BJMST08a}
for the special case of the ground state correlation functions
of the infinite chain. As is implicit in \cite{JMS08} the
JMS theorem implies the validity of such a formula in a much
more general situation. This had been conjectured at an early stage
of the development of the method in \cite{BGKS07} and was partially
explored in \cite{BDGKSW08,TGK10a}. The difference between these
older works and the present work is that the construction of the
finite building blocks of the annihilation operators in the old
work was based on inhomogeneous monodromy matrices, while
we will be using homogeneous monodromy matrices here. This
has the advantage that we have to process smaller expressions
allowing us a more efficient use of our computers.

Define the $\k$-trace $\tr^\k: {\cal W}_\a \rightarrow {\mathbb C}$
by
\begin{equation} \label{defkappatrace}
     \tr^\k \bigl\{X^{(\a)}\bigr\} =
        \frac{\tr_{[k,l]} \bigl\{q^{- \k S_{[k,l]}} X_{[k,l]}\bigr\}}
             {\tr_{[k,l]} \bigl\{q^{- \k S_{[k,l]}}\bigr\}} \epc
\end{equation}
where $X^{(\a)} \in {\cal W}_\a$ with non-trivial part
$X_{[k,l]}$, and a function
\begin{equation} \label{defomzero}
     \om_0 (\z; \a) =
        - \biggl(\frac{1 - q^\a}{1 + q^\a}\biggr)^2 \D_\z \ps(\z, \a) \epp
\end{equation}
Here we have introduced the notation $\D_\z f (\z) = f(q \z)
- f(q^{-1} \z)$. Further, let
\begin{equation} \label{defomop}
     \Om = \int_\G \frac{\rd \z_1^2}{2 \p \i \z_1^2}
           \int_\G \frac{\rd \z_2^2}{2 \p \i \z_2^2}
	   \bigl(\om_0 (\z_1/\z_2; \a) - \om(\z_1, \z_2; \a)\bigr)
	   \bv (\z_1) \cv (\z_2) \epc
\end{equation}
where $\G$ is a simple closed contour around $1$. Then it is
possible to infer from the JMS theorem that, for any $X^{(\a)}
\in {\cal W}_{\a, 0}$,
\begin{equation} \label{relztralpha}
     Z^{- \a/2} \bigl\{X^{(\a)}\bigr\} =
          \tr^\a \bigl\{\re^\Om X^{(\a)}\bigr\} \epp
\end{equation}
It follows that for any realisation of the functional $Z^0$
and any local operator $X \in {\cal W}$
\begin{equation} \label{expkappanull}
     Z^0 \bigl\{X\bigr\} = \lim_{\a \rightarrow 0}
          \tr^\a \bigl\{\re^\Om X^{(\a)}\bigr\} \epp
\end{equation}
This is the formula that we shall use in order to calculate
short range correlation functions below.

We shall argue below that (\ref{expkappanull}) remains
valid for $\kappa \ne 0$ if we restrict the class of operators
to those that are invariant under spin reversal. We define
the spin-reversal operator on $\End {\cal H}_{[k,l]}$ by
${\mathbb J}_{[k,l]} X_{[k,l]} = J_{[k,l]}  X_{[k,l]}
J_{[k,l]}$, $J = \prod_{j = k}^l \s_j^x$. The corresponding
spin-reversal operator on ${\cal W}$ is denoted
${\mathbb J}$. $X \in {\cal W}$ is spin-reversal invariant
if ${\mathbb J} X = X$. We claim that for such operators
\begin{equation} \label{expkappa}
     Z^\k \bigl\{X\bigr\} = \lim_{\a \rightarrow 0}
          \tr^\a \bigl\{\re^\Om X^{(\a)}\bigr\} \epp
\end{equation}
Here the $\k$-dependence on the right hand side is hidden
in $\Om$ which depends on $\k$ through the function $\om$.

The paper is organized as follows. In Section~\ref{sec:modeconstruction}
we explain the construction of the modes from finite building
blocks. Starting point are two operators $\kv_{[k,l]}$ and
$\tv^*_{[k,l]}$ acting on finite chains and their partial
fraction decompositions. From these we obtain finite chain
versions of the creation and annihilation operators which we
present explicitly in terms of the Laurent coefficients of
$\kv_{[k,l]}$ and $\tv^*_{[k,l]}$ and in terms of coefficients
characterizing the behaviour of these operators in the limit
when the spectral parameter goes to infinity. We also discuss
the reduction properties of the operators which allows us
to construct operators for the infinite chain from finite
building blocks. In Section~\ref{sec:cf} we apply the formalism
to work out the finite-temperature two-point correlation functions
of the XXZ chain on up to five lattice sites. We discuss these
functions and compare with known asymptotic results.
Section~\ref{sec:concl} is devoted to a short summary and to
our conclusions. In Appendix~\ref{app:jms_to_expform} we obtain
the exponential form starting from the JMS theorem.
Appendix~\ref{app:multtoadd} connects two different forms of
the physical part $\r$, $\om$ with multiplicative and with
additive spectral parameters. Appendix~\ref{app:n5_explicit}
lists the expressions for the fourth-neighbour two-point
functions in terms of $\r$ and $\om$.

\section{Construction of the modes from finite building blocks}
\label{sec:modeconstruction}
\subsection{\boldmath $L$-matrices and monodromy matrices}
The creation and annihilation operators introduced in the previous
section are constructed from weighted traces of the elements of
certain monodromy matrices related to $U_q (\widehat{\mathfrak{sl}_2})$.
These monodromy matrices are products of two types of $L$-matrices with
two-dimensional or infinite-dimensional auxiliary space.

The $L$-matrices with two-dimensional auxiliary space are directly
related to the $R$-matrix of the six-vertex model,
\begin{align}
     & R(\zeta) = \begin{pmatrix}
		     1&0&0&0\\
		     0&\be(\z)&\g(\z)&0\\
		     0&\g(\z)&\be(\z)&0\\
		     0&0&0&1
		  \end{pmatrix} \epc
\end{align}
where
\begin{equation}
     \be(\z) = \frac{(1 - \z^2)q}{1 - q^2 \z^2} \epc \qd
     \g(\z) = \frac{(1 - q^2)\z}{1 - q^2 \z^2} \epp
\end{equation}
Fixing an auxiliary space $V_a \cong {\mathbb C}^2$ we define
$L_{a,j} (\z) = \varrho (\z) R_{a,j} (\z) \in \End\bigl(V_a \otimes
{\cal H}_{[k,l]}\bigr)$, where $\varrho (\z)$ is a scalar factor
to be specified below. This is the standard $L$-matrix of the
six-vertex model. The corresponding monodromy matrix is
\begin{equation}
     T_{a,[k,l]} (\z) = L_{a,l} (\z) \dots L_{a,k} (\z) \epp
\end{equation}
We are going to construct operators acting on $\End {\cal H}_{[k,l]}$
or on $\End\bigl(V_a \otimes {\cal H}_{[k,l]}\bigr)$. Our first
example is the $L$-operator ${\mathbb L}_{a,j} (\z)$ defined by
\begin{equation}
     {\mathbb L}_{a,j} (\z)\, \id_a \otimes X_{[k,l]}
        = L_{a,j} (\z) \id_a \otimes X_{[k,l]} L_{a,j}^{-1} (\z)
\end{equation}
for $j \in \{k, \dots, l\}$. These $L$-operators generate the
first type of monodromy matrices needed for the construction
of the Fermionic operators,
\begin{equation} \label{monofund}
     {\mathbb T}_{a, [k,l]} (\z, \a) \, \id_a \otimes X_{[k,l]} =
        {\mathbb L}_{a,l} (\z) \dots {\mathbb L}_{a,k} (\z) \,
	q^{\a \s_a^z} \id_a \otimes X_{[k,l]} \epp
\end{equation}

A second type of monodromy matrix acts adjointly on an infinite
dimensional module of the $q$-oscillator algebra $\rm Osc$ which is
defined in terms of generators $\av$, $\av^*$, $q^{\pm D}$ that
satisfy the relations
\begin{equation}
     q^D \av q^{-D} = q^{-1} \av \epc \qd
     q^D \av^* q^{-D} = q \, \av^* \epc \qd
     \av \av^* = 1 - q^{2D + 2} \epc \qd
     \av^* \av = 1 - q^{2D} \epp
\end{equation}
We fix a copy ${\rm Osc}_A$ of ${\rm Osc}$ and define
\begin{equation} \label{Losc}
     L_{A,j} (\z) = \s (\z)
                    \begin{pmatrix}
		    1 - \z^2 q^{2D_A + 2} & - \z \av_A \\
		    - \z \av^*_A & 1
                    \end{pmatrix}_j
                    \begin{pmatrix}
		    q^{- D_A} & 0 \\ 0 & q^{D_A}
                    \end{pmatrix}_j
\end{equation}
which acts on ${\rm Osc}_A \otimes {\cal H}_{[k,l]}$. The scalar
factor $\s(\z)$ is a convenient normalization and is required to
solve the functional equation
\begin{equation}
     \s (\z) \s (q^{-1} \z) = \frac{1}{1 - \z^2} \epp
\end{equation}
It also fixes the scalar factor in the definition of $L_{a,j} (\z)$,
\begin{equation}
     \varrho (\z) = \frac{q^{-1/2} \s(q^{-1} \z)}{\s(\z)} \epp
\end{equation}
For the $L$-operator (\ref{Losc}) we define its adjoint action on
$\End \bigl({\rm Osc}_A \otimes {\cal H}_{[k,l]}\bigr)$
\begin{equation}
     {\mathbb L}_{A,j} (\z)\, \id_A \otimes X_{[k,l]}
        = L_{A,j} (\z) \id_A \otimes X_{[k,l]} L_{A,j}^{-1} (\z) \epc
\end{equation}
$j \in \{k, \dots, l\}$, and the corresponding monodromy matrix
\begin{equation} \label{monoosc}
     {\mathbb T}_{A, [k,l]} (\z, \a) \, \id_A \otimes X_{[k,l]} =
        {\mathbb L}_{A,l} (\z) \dots {\mathbb L}_{A,k} (\z) \,
	q^{2 \a D_A} \id_A \otimes X_{[k,l]} \epp
\end{equation}
The monodromy matrices (\ref{monofund}) and (\ref{monoosc}) are the
basic ingredients in the definition of the Fermionic operators.

\subsection{\boldmath The basic operators on finite chain segments}
The Fermionic operators and their mode expansions (\ref{modes})
are constructed from a parental operator $\kv_{[k,l]} (\z, \a)
\in \End \End {\cal H}_{[k,l]}$ and from an adjointly acting transfer
matrix $\tv^*_{[k,l]} (\z, \a) \in \End \End {\cal H}_{[k,l]}$
which were both introduced in \cite{BJMST08a}. Before recalling
their definitions we have to fix some notation. Throughout this
section we shall fix an operator $X_{[k,l]} \in \End {\cal H}_{[k,l]}$
of spin $s$ and length $n \le k - l + 1$. For such an operator we
necessarily have that $|s| \le n$. If $\xv_{[k,l]} \in \End \End
{\cal H}_{[k,l]}$ is acting on $X_{[k,l]}$, we omit the subscript
and write $\xv X_{[k,l]} = \xv_{[k,l]} X_{[k,l]}$ for simplicity.

With this convention the parental operator $\kv_{[k,l]} (\z, \a)$
is defined by
\begin{equation}
\label{eq:k}
     \kv (\z, \a) \, X_{[k,l]}
        = \tr_{A, a} \Bigl\{ \s_a^+ {\mathbb T}_a (\z, \a)
	             {\mathbb T}_A (\z, \a) \z^{\a - {\mathbb S}}
		     \bigl(\id_{A, a} \otimes
		     q^{- 2S^z_{[k,l]}} X_{[k,l]}\bigr) \Bigr\} \epp
\end{equation}
The trace over the oscillator algebra is to be understood in
such a way that it agrees with the trace over the module $W^+
= \oplus_{k \ge 0} {\mathbb C} |k\>$ for $|q| < 1$ (with
$q^D |k\> = q^k |k\>, \av |k\> = (1 - q^{2k}) |k - 1\>, \av^* |k\> =
|k + 1\>$). For details see appendix A of \cite{BJMST08a}.

The operator $\tv_{[k,l]}^* (\z, \a)$ is defined by
\begin{equation}
     \tv^* (\z, \a) \, X_{[k,l]}
        = \tr_a \Bigl\{{\mathbb T}_a (\z, \a)
		       \bigl(\id_a \otimes X_{[k,l]}\bigr) \Bigr\} \epp
\end{equation}

\subsection{\boldmath Partial fraction decomposition
of $\kv_{[k,l]} (\z, \a)$ and the annihilation operators}
In the definition of the monodromy operators (\ref{monofund}),
(\ref{monoosc}) we could have included an inhomogeneity parameter
on each site if we would have replaced $\z$ by $\z/\x_j$ in
every factor on the right hand side of these equations (cf.\
\cite{BJMST08a}). This would have affected the analytic
properties of the operators $\kv_{[k,l]} (\z, \a)$ and
$\tv_{[k,l]}^* (\z, \a)$ which have only simple poles in
the inhomogeneous case, if the inhomogeneities are mutually
distinct, but exhibit poles of higher order in the homogeneous
case. In our previous work \cite{BDGKSW08} we considered the
inhomogeneous case. Here we restrict ourselves to the
homogeneous case. This means that we have to deal with slightly
more complicated expressions which depend, however, on less
parameters. The latter fact renders them more efficient in
computer algebraic calculations.

We will express the Fermionic annihilation
operator $\cv_{[k,l]} (\z, \a)$ and some other operators which are
important in the construction of the creation operators in terms of the
Laurent coefficients of the operator $\kv_{\rm skal\ [k,l]} (\z, \a)$
defined by $\kv_{\rm skal} (\z, \a) X_{[k,l]} = \z^{- \a - s - 1}
\kv (\z, \a) X_{[k,l]}$. As can be seen from section 2.5 of
\cite{BJMST08a},  $\kv_{\rm skal} (\z, \a) X_{[k,l]}$ is a rational
function in $\z^2$. As a function of $\z^2$ it has three $n$-fold
poles at $q^{2 \eps}$, $\eps = 0, \pm 1$, and an $s$-fold pole at
$0$ if $s > 0$. From these properties and the asymptotic behaviour
we obtain the partial fraction decomposition
\begin{equation} \label{kskalexp}
     \kv_{\rm skal} (\z, \a) X_{[k,l]}
        = \biggl[ \sum_{j=1}^n \sum_{\eps = - 1}^1 
	  \frac{\r_j^{(\eps)} (\a)}{(\z^2 - q^{2 \eps})^j}
        + \sum_{j=1}^s \frac{\k_j (\a)}{\z^{2j}} \biggr] X_{[k,l]}  \epp
\end{equation}

The definitions of the operators $\bar \cv_{[k,l]} (\z, \a)$,
$\cv_{[k,l]} (\z, \a)$ and $\fv_{[k,l]} (\z, \a)$ in \cite{BJMST08a}
remain valid in the homogeneous case. They are introduced in such
a way that the decomposition
\begin{equation}
      \kv (\z, \a) X_{[k,l]}
         = \bigl(\bar \cv(\z, \a) + \cv(q \z, \a) + \cv(q^{-1} \z, \a)
	   + \fv(q \z, \a) - \fv(q^{-1} \z, \a)\bigr)  X_{[k,l]}
\end{equation}
holds, which becomes unique if we fix
\begin{subequations}
\begin{align}
     \bar \cv (\z, \a) X_{[k,l]} &
        = \int_\G \frac{\rd \x^2}{2 \p \i} \ps(\z/\x, \a + s + 1)
	          \kv (\x, \a) X_{[k,l]} \epc \\[1ex] \label{defcint}
     \cv (\z, \a) X_{[k,l]} &
        = \int_\G \frac{\rd \x^2}{4 \p \i} \ps(\z/\x, \a + s + 1)
	          \bigl(\kv (q \x, \a) + \kv (q^{-1} \x, \a)\bigr) X_{[k,l]} \epc \\[1ex]
     \fv (\z, \a) X_{[k,l]} &
        = \bigl(\fv^{\rm sing} (\z, \a) + \fv^{\rm reg} (\z, \a)\bigr) X_{[k,l]} \epc \\[1ex]
     \fv^{\rm sing} (\z, \a) X_{[k,l]} &
        = \int_\G \frac{\rd \x^2}{4 \p \i} \ps(\z/\x, \a + s + 1)
	          \bigl(- \kv (q \x, \a) + \kv (q^{-1} \x, \a)\bigr) X_{[k,l]} \epp
\end{align}
\end{subequations}
Here $\G$ is a small circle around $\x^2 = 1$. It is not difficult
to evaluate the above integrals and to obtain the operators
$\bar \cv_{[k,l]} (\z, \a)$, $\cv_{[k,l]} (\z, \a)$ and
$\fv^{\rm sing}_{[k,l]} (\z, \a)$ explicitly in terms of the
Laurent coefficients $\r_j^{(\eps)} (\a)$ and $\k_j (\a)$ occurring
in the partial fraction decomposition of $\kv_{\rm skal, [k,l]}
(\z, \a)$. We have to deal with integrands of the form
\begin{multline}
     \x^{-2} \ps(\z/\x, \a + s + 1) \kv (\h, \a) X_{[k,l]} \\
        = \z^{\a + s + 1} \, \frac{\z^2 + \x^2}{2 \x^2 (\z^2 - \x^2)}
	  \biggl( \frac{\h}{\x} \biggr)^{\a + s + 1} \kv_{\rm skal} (\h, \a) X_{[k,l]}
	  \epp
\end{multline}
with $\h = \x q^\eps$, $\eps = - 1, 0, 1$. Using the decomposition
\begin{equation} \label{psidetox}
     \frac{\z^2 + \x^2}{2 \x^2 (\z^2 - \x^2)} = 
        \2 \sum_{k=0}^\infty (-1)^k (\x^2 - 1)^k + \frac{1}{\z^2 - 1}
	   \sum_{k=0}^\infty \biggl( \frac{\x^2 - 1}{\z^2 - 1} \biggr)^k \epc
\end{equation}
valid for $|\x^2 - 1| < |\z^2 - 1| < 1$, the integrals are easily
calculated. The remaining part $\fv^{\rm reg}_{[k,l]} (\z, \a)$
is obtained by inverting the difference operator $\D_\z$ on
the space of Laurent polynomials in $\z^2$.
\begin{lemma}
The operators $\bar \cv_{[k,l]} (\z, \a)$, $\cv_{[k,l]} (\z, \a)$ and
$\fv_{[k,l]} (\z, \a)$ have the partial fraction decompositions
\begin{subequations}
\begin{align}
     \bar \cv(\z,\a) X_{[k,l]} & =
        \z^{\a + s + 1} \sum_{j=0}^n \frac{\bar c_j (\a)}{(\z^2 - 1)^j} X_{[k,l]}
	\epc \\[1ex] \label{claurant}
     \cv(\z,\a) X_{[k,l]}  & =
        \z^{\a + s + 1} \sum_{j=0}^n \frac{c_j (\a)}{(\z^2 - 1)^j} X_{[k,l]}
	\epc \\[1ex] \label{fop}
     \fv(\z,\a) X_{[k,l]} & =
        \z^{\a + s + 1} \biggl[ \sum_{j=0}^n \frac{f_j (\a)}{(\z^2 - 1)^j}
        + \sum_{j=0}^s \frac{\k_j (\a) \z^{-2j}}
	                    {q^{\a + s + 1 - 2j} - q^{2j - \a - s - 1}}
			       \biggr] X_{[k,l]}  \epc
\end{align}
\end{subequations}
where the coefficients are determined in terms of the coefficients of
the partial fraction decomposition of $\kv_{{\rm skal}, [k,l]} (\z, \a)$,
\begin{subequations}
\label{tcmodes}
\begin{align}
     \bar c_j (\a) & = \r_j^{(0)} (\a) \epc \qd j = 1, \dots, n \epc \\
     \bar c_0 (\a) & = \2 \sum_{j=1}^n (-1)^{j-1} \bar c_j (\a) \epc \\
     c_j (\a) & = \2
        \bigl( q^{2j - \a - s - 1} \r_j^{(-)} (\a) +
               q^{\a + s + 1 - 2j} \r_j^{(+)} (\a) \bigr)
        \epc \qd j = 1, \dots, n \epc \label{cmodes} \\[1ex]
     c_0 (\a) & = \2 \sum_{j=1}^n (-1)^{j-1} c_j (\a) \epc \label{cnull} \\
     f_j (\a) & = \2
        \bigl( q^{2j - \a - s - 1} \r_j^{(-)} (\a) -
               q^{\a + s + 1 - 2j} \r_j^{(+)} (\a) \bigr)
        \epc \qd j = 1, \dots, n \epc \\[1ex]
     f_0 (\a) & = \2 \sum_{j=1}^n (-1)^{j-1} f_j (\a) \epc \\
     \k_0 (\a) & = \2 \sum_{j=1}^n (-1)^j
        \sum_{\eps = - 1}^1 q^{- 2 \eps j} \r_j^{(\eps)} (\a) \epp
\end{align}
\end{subequations}
\end{lemma}

The operator $\bv_{[k,l]} (\z,\a)$ that is needed in the construction of
the density matrix is defined as
\begin{equation}
     \bv_{[k,l]} (\z,\a) =
        - q^{-1}(q^{\a - \spin_{[k,l]} - 1} - q^{- \a + \spin_{[k,l]} + 1})
	  \spinflip_{[k,l]} \cv_{[k,l]} (\z,- \a) \spinflip_{[k,l]} \epp
\end{equation}
Setting
\begin{equation} \label{bmodes}
     b_j (\a) = - q^{-1} (q^{\a - \spin_{[k,l]} - 1} - q^{- \a + \spin_{[k,l]} + 1})
                  \spinflip_{[k,l]} c_j (-\a) \spinflip_{[k,l]}
\end{equation}
we obtain it in the form
\begin{equation} \label{blaurant}
     \bv(\z,\a) X_{[k,l]}  =
        \z^{- \a - s + 1} \sum_{j=0}^n \frac{b_j (\a)}{(\z^2 - 1)^j} X_{[k,l]} \epp
\end{equation}

All one has to do in order to obtain explicit representations for
these operators is to extract the coefficients $\r_j^{(\eps)} (\a)$
and $\k_j (\a)$ from the expansion (\ref{kskalexp}). They can be
obtained as
\begin{equation}
\label{eq:laurent_rho}
     \r_{n-k}^{(\eps)} (\a) =
        \lim_{\z^2 \rightarrow q^{2 \eps}}
        \frac{1}{k!} \6_{\z^2}^k (\z^2 - q^{2 \eps})^n
	   \kv_{\rm skal\ [k,l]} (\z, \a)
\end{equation}
for $k = 0, \dots, n - 1$, and
\begin{equation}
\label{eq:laurent_kappa}
     \k_{s-k}^{(\eps)} (\a) =
        \lim_{\z^2 \rightarrow 0}
        \frac{1}{k!} \6_{\z^2}^k \z^{2s} \kv_{\rm skal\ [k,l]} (\z, \a) \epc
\end{equation}
where $k = 0, \dots, s - 1$.
\subsection{\boldmath Partial fraction decomposition of the operator
$\tv^* (\z, \a)$ in the homogeneous case}
In section 2.5 of \cite{BJMST08a} it is shown that $\tv^*_{[k,l]} (\z, \a)$
is a rational function in $\z^2$. Directly from its definition
we further find that for $\z^2 \rightarrow \infty$
\begin{equation}
     \tv^* (\z, \a) X_{[k,l]} \sim (q^{\a + s} + q^{- \a - s}) X_{[k,l]}
\end{equation}
and that this operator has precisely two $n$-fold poles at
$\z^2 = q^{\pm 2}$. Thus, it has the partial fraction decomposition
\begin{equation} \label{tstarparfrac}
     \tv^* (\z, \a) X_{[k,l]} = \biggl[ \sum_{j=1}^n \sum_{\eps = \pm}
	\frac{\tau_j^{(\eps)} (\a)}{(\z^2 - q^{\eps 2})^j}
        + q^{\a + s} + q^{- \a - s} \biggr] X_{[k,l]} \epp
\end{equation}

\subsection{Construction of the finite generators of creation operators}
The construction of creation operators is intimately connected with
the Taylor expansion of $\tv^*_{[k,l]} (\z, \a)$ around $\z^2 = 1$.
\begin{lemma}
Taylor expanding (\ref{tstarparfrac}) around $\z^2 = 1$ we obtain
\begin{equation} \label{tstartaylor}
     \tv^* (\z, \a) X_{[k,l]}
        = \sum_{p=1}^\infty (\z^2 - 1)^{p-1} t^*_p (\a) \, X_{[k,l]} \epc
\end{equation}
where
\begin{subequations}
\label{tstarmodes}
\begin{align}
     & t_1^* (\a) X_{[k,l]} = \biggl[
        \sum_{j=1}^n \sum_{\eps = \pm}
	\frac{\tau_j^{(\eps)} (\a)}{(1 - q^{\eps 2})^j} + q^{\a + s} + q^{- \a - s}
	\biggr] X_{[k,l]} \epc \\
     & t_p^* (\a) X_{[k,l]} =
        \sum_{j=1}^n \sum_{\eps = \pm} \binom{-j}{p-1}
	\frac{\tau_j^{(\eps)} (\a) X_{[k,l]}}{(1 - q^{\eps 2})^{p + j - 1}} \epc
	\qd p = 2, 3, \dots
\end{align}
\end{subequations}
\end{lemma}

We now proceed with the construction of the Fermionic creation operators.
In equation (2.29) of \cite{BJMST08a} their finite generators
$\bv^*_{[k,l]} (\z, \a)$ are defined by
\begin{equation} \label{defbstar}
     \bv^* (\z, \a) X_{[k,l]}
        = \bigl( \fv (q \z, \a) + \fv (q^{-1} \z, \a)
	         - \tv^* (\z, \a) \fv (\z, \a) \bigr) X_{[k,l]} \epp
\end{equation}
Here $\fv_{[k,l]} (\z, \a)$ is given by (\ref{fop}) and
$\tv^*_{[k,l]} (\z, \a)$ by (\ref{tstartaylor}).

It is shown in \cite{BJMST08a} (see Lemma 3.8) that the operator
$\bv^*_{[k,l]} (\z, \a)$ is regular at $\z^2 = 1$. Inserting
(\ref{fop}) and (\ref{tstartaylor}) into (\ref{defbstar}) regularity
implies the identities
\begin{equation}
     \sum_{j = - p + 1}^{n} t^*_{p+j} (\a) f_j (\a) = 0
        \qd \text{for $p = - n + 1, \dots, 0$} \epp
\end{equation}
Using this identity in the Taylor expansion of the right hand side
of (\ref{defbstar}) around $\z^2 = 1$ we obtain the following
\begin{lemma}
\begin{equation} \label{bstartaylor}
     \bv^* (\z, \a) X_{[k,l]}
        = \z^{\a + s + 1} \sum_{p=1}^\infty (\z^2 - 1)^{p-1} b^*_p (\a) \, X_{[k,l]} \epc
\end{equation}
where
\begin{multline} \label{bstarmodes}
     b_p^* (\a) X_{[k,l]}
        = \biggl[ \sum_{j=0}^n \binom{-j}{p-1}
	  \frac{q^{\a + s + p - j} - (-1)^{p+j} q^{j - p - \a - s}}
	       {(q - q^{-1})^{p + j -1}} \, f_j (\a) \\
          + \sum_{j=0}^s \binom{-j}{p-1}
	  \frac{q^{\a + s + 1 - 2j} + q^{2j - \a - s - 1}}
	       {q^{\a + s + 1 - 2j} - q^{2j - \a - s - 1}} \, \k_j (\a) \\
          - \sum_{j=0}^n t_{p+j}^* (\a) f_j (\a)
	  - \sum_{j=0}^s \sum_{k=0}^{p-1} \binom{-j}{k}
	    \frac{t_{p-k}^* (\a) \k_j (\a)}
	         {q^{\a + s + 1 - 2j} - q^{2j - \a - s - 1}} \biggr] X_{[k,l]}
		 \epp
\end{multline}
\end{lemma}

The operator $\cv^*_{[k,l]} (\z,\a)$ is defined as
\begin{equation}
     \cv^*_{[k,l]} (\z,\a) =
        q^{-1}(q^{\a - \spin_{[k,l]} - 1} - q^{- \a + \spin_{[k,l]} + 1})
	\spinflip_{[k,l]} \bv^*_{[k,l]} (\z,- \a) \spinflip_{[k,l]} \epp
\end{equation}
Setting
\begin{equation} \label{cstarmodes}
     c_p^* (\a) = q^{-1} (q^{\a - \spin_{[k,l]} - 1} - q^{- \a + \spin_{[k,l]} + 1})
                  \spinflip_{[k,l]} b_p^* (-\a) \spinflip_{[k,l]}
\end{equation}
we obtain the Taylor expansion
\begin{equation} \label{cstartaylor}
     \cv^* (\z,\a) X_{[k,l]}  =
        \z^{- \a - s + 1} \sum_{p=1}^\infty (\z^2 - 1)^{p-1} c_p^* (\a) X_{[k,l]} \epp
\end{equation}

In the following subsection we will use certain reduction properties
of the operators acting on finite segments of the infinite spin chain
in order to construct operators that act on the whole infinite chain.
\subsection{Reduction relations and extension of the action to infinite chains}
The family of operators acting on $\End {\cal H}_{[j,m]}$
that was introduced above has certain simple so-called `reduction
properties' describing how their action on operators of the form
$q^{\a (\s_j^z + \dots + \s_{k-1}^z)} \id_{[j,k-1]} \otimes
X_{[k,l]} \otimes \id_{[l+1,m]}$, where $[k,l] \subset [j,m]$
is a proper subset, reduces to the action of `shorter operators'.
The operators $\kv_{[k,l]}$ and $\tv^*_{[k,l]}$ inherit the
following left reduction relations from the `gauge invariance'
of the universal $R$-matrix
\begin{align}
     \tv^*_{[k-1,l]} (\z,\a) \bigl(q^{\a \s^z} \otimes X_{[k,l]}\bigr) & =
        q^{\a \s^z} \otimes \bigl(\tv^*_{[k,l]} (\z,\a) X_{[k,l]}\bigr) \epc \\[1ex]
     \kv_{[k-1,l]} (\z,\a) \bigl(q^{(\a + 1) \s^z} \otimes X_{[k,l]}\bigr) & =
        q^{\a \s^z} \otimes \bigl(\kv_{[k,l]} (\z,\a) X_{[k,l]}\bigr) \epp
\end{align}
These properties can be easily understood by means of the graphical
representation of the operators. As an immediate consequence we obtain
the left reduction relations
\begin{equation} \label{leftredubyone}
     \xv_{[k-1,l]} (\z,\a) \bigl(q^{(\a + s(\xv)) \s^z} \otimes X_{[k,l]}\bigr) =
        q^{\a \s^z} \otimes \bigl(\xv_{[k,l]} (\z,\a) X_{[k,l]} \bigr) \epc
\end{equation}
where $\xv = \bv, \bv^*, \cv, \cv^*, \tv^*$, and where $s(\xv)$
denotes the spin of $\xv$ which, as we recall, is $-1$ for $\bv$,
$\cv^*$, $+1$ for $\cv$, $\bv^*$ and $0$ for $\tv^*$.

The left reduction relations can be used to extend the action of
the operators $\xv_{[k,l]} (\z,\a)$ on products of local operators
with tail inductively to the semi-infinite interval $(- \infty, l]$.
Iterating (\ref{leftredubyone}) we obtain
\begin{equation} \label{allleft}
     \xv_{(-\infty,l]} (\z,\a)
        \bigl(q^{2(\a + s(\xv)) S_{(-\infty,k-1]}} \otimes X_{[k,l]}\bigr) =
        q^{2\a S_{(-\infty,k-1]}} \otimes \bigl(\xv_{[k,l]} (\z,\a) X_{[k,l]} \bigr) \epp
\end{equation}
An extension of these operators to the right is less obvious and
is essentially different for the creation and the annihilation operators.

The annihilation operators have the property that they do not enlarge,
in particular not to the right, the support of an operator $X_{[k,l]}$
they are acting on,
\begin{equation} \label{annihiliredu}
     \xv_{[k,l+m]} (\z,\a) \bigl(X_{[k,l]} \otimes \id_{[l+1,l+m]}\bigr)
        = \bigl(\xv_{[k,l]} (\z,\a) X_{[k,l]}\bigr) \otimes \id_{[l+1,l+m]}
\end{equation}
for $\xv = \bv, \cv$ and $m \in {\mathbb N}$. The proof is non-trivial
and can be found in section 3.5 of \cite{BJMST08a}. Since
(\ref{annihiliredu}) is valid for every $m \in {\mathbb N}$ we may
use this equation to extend the action of the operators $\xv = \bv, \cv$
to a semi-infinite interval, setting
\begin{equation} \label{anrredu}
     \xv_{[k,\infty)} (\z,\a) \bigl(X_{[k,l]} \otimes \id_{[l+1,\infty)}\bigr)
        = \bigl(\xv_{[k,l]} (\z,\a) X_{[k,l]}\bigr) \otimes \id_{[l+1,\infty)} \epp
\end{equation}

By way of contrast the creation operators extend
the support of any operator $X_{[k,l]}$ indefinitely to the right.
Yet, interestingly, their modes do not. Every mode $x_p$, $x = b^*, c^*,
t^*$, in the expansions (\ref{tstartaylor}), (\ref{bstartaylor}) and
(\ref{cstartaylor}) extends the length of an operator $X_{[k,l]}$ at
most by $p$ to the right. For $\tv^*_{[k,l]}$ a relatively simple and
elegant proof of this fact is presented in section 3.4 of \cite{BJMST08a}.
The authors derive an expansion of the form
\begin{multline}
     \tv^*_{[k,l+m]} (\z,\a) \bigl(X_{[k,l]} \otimes \id_{[l+1,l+m]}\bigr) \\ =
        \sum_{p=1}^{m} \bigl(\yv_{[k,l+p]}^{(p)} (\z,\a)
	                     (X_{[k,l]} \otimes \id_{[l+1,l+p]})\bigr)
			     \otimes \id_{[l+p+1,l+m]} \\[-1ex]
	+ \zv_{[k,l+m]}^{(m)} (\z,\a)
	  \bigl(X_{[k,l]} \otimes \id_{[l+1,l+m]} \bigr) \epp
\end{multline}
From the explicit form of the operators $\yv_{[k,l+p]}^{(p)}$ and
$\zv_{[k,l+m]}^{(m)}$ we can read off the following features. The
$\yv_{[k,l+p]}^{(p)} (\z,\a) \bigl(X_{[k,l]} \otimes \id_{[l+1,l+p]}
\bigr)$ have support $[k,l+p]$, and $\zv_{[k,l+m]}^{(m)} (\z,\a)
\bigl(X_{[k,l]} \otimes \id_{[l+1,l+m]}\bigr)$ has support $[k,l+m]$.
These operators depend rationally on $\z^2$. Close to $\z^2 = 1$ they
behave as $\yv_{[k,l+p]}^{(p)} (\z,\a) = \CO \bigl( (\z^2 - 1)^{p-1} \bigr)$
and $\zv_{[k,l+m]}^{(m)} (\z,\a) = \CO \bigl( (\z^2 - 1)^m \bigr)$.
This implies that $\yv_{[k,l+r]}^{(r)} (\z,\a) \bigl(X_{[k,l]} \otimes
\id_{[l+1,l+m]}\bigr)$ for $r > p$ does not contribute to the first
$p$ terms of the Taylor expansion (\ref{tstartaylor}) of
$\tv^*_{[k,l+r]} (\z,\a) \bigl(X_{[k,l]} \otimes \id_{[l+1,l+m]}\bigr)
$ in $\z^2$ around $1$. In other words the $p$th coefficient
$t_{p, [k,l+m]}^* (\a)$ of the Taylor expansion extends the support
at most by $p$ to the right,
\begin{multline} \label{tstarmodesredu}
     t_{p, [k,l+m]}^* (\a) \bigl(X_{[k,l]} \otimes \id_{[l+1,l+m]} \bigr) \\ 
        = \bigl(t_{p, [k,l+p]}^* (\a) (X_{[k,l]} \otimes \id_{[l+1,l+p]})\bigr)
	  \otimes \id_{[l+p+1,l+m]}
\end{multline}
for $p = 1, \dots, m$. We may also read this as being valid for
fixed $p$ and every $m \ge p$. Since the first factor on the right
hand side of this equation is independent of $m$, we may use
(\ref{tstarmodesredu}) to extend the action of the modes infinitely
to the right, setting
\begin{equation} \label{treduright}
     t_{p, [k,\infty)}^* (\a) \bigl(X_{[k,l]} \otimes \id_{[l+1,\infty)}\bigr)
        = \bigl(t_{p, [k,l+p]}^* (\a) (X_{[k,l]} \otimes \id_{[l+1,l+p]})\bigr)
	  \otimes \id_{[l+p+1,\infty)} \epp
\end{equation}

As for the Fermionic creation operators it follows from (\ref{cstarmodes})
that the modes of $\cv^*$ have the same reduction relations as those
of $\bv^*$. The latter can be obtained from a similar argument as
above by combining Lemma 3.1 and Lemma 3.7 of \cite{BJMST08a}. These
lemmata imply that
\begin{multline}
     \bv^*_{[k,l+m]} (\z,\a) \bigl(X_{[k,l]} \otimes \id_{[l+1,l+m]}\bigr) \\ =
        \z^{\a + s +1} \Bigl[ \sum_{p=1}^{m} \bigl(\uv_{[k,l+p]}^{(p)} (\z,\a)
	   (X_{[k,l]} \otimes \id_{[l+1,l+p]}\bigr) \otimes \id_{[l+p+1,l+m]} \\[-1ex]
	+ \vv_{[k,l+m]}^{(m)} (\z,\a) \bigl(X_{[k,l]} \otimes \id_{[l+1,l+m]}\bigr)
	  \Bigr] \epc
\end{multline}
where $\uv_{[k,l+p]}^{(p)} (\z,\a)$, $p = 1, \dots, m$, and
$\vv_{[k,l+m]}^{(m)} (\z,\a)$ have the same properties as the
operators $\yv_{[k,l+p]}^{(p)} (\z,\a)$ and $\zv_{[k,l+m]}^{(m)} (\z,\a)$
introduced above. A comparison with (\ref{bstartaylor}) then implies
that
\begin{multline} \label{bstarmodesredu}
     b_{p, [k,l+m]}^* (\a) \bigl(X_{[k,l]} \otimes \id_{[l+1,l+m]}\bigr) \\
        = \bigl(b_{p, [k,l+p]}^* (\a) (X_{[k,l]} \otimes \id_{[l+1,l+p]})\bigr)
	  \otimes \id_{[l+p+1,l+m]}
\end{multline}
for all $m \ge p$, and therefore also
\begin{multline} \label{cstarmodesredu}
     c_{p, [k,l+m]}^* (\a) \bigl(X_{[k,l]} \otimes \id_{[l+1,l+m]}\bigr) \\
        = \bigl(c_{p, [k,l+p]}^* (\a) (X_{[k,l]} \otimes \id_{[l+1,l+p]})\bigr)
	  \otimes \id_{[l+p+1,l+m]}
\end{multline}
for all $m \ge p$. Eqs.\ (\ref{bstarmodesredu}) and (\ref{cstarmodesredu})
allow us to extend the action of the modes on local operators
infinitely to the right (or, turning it the other way round, to
define the action of every mode on a local operator in terms of
finite matrices),
\begin{multline} \label{crreduright}
     x_{p, [k,\infty)}^* (\a) \bigl(X_{[k,l]} \otimes \id_{[l+1,\infty)}\bigr) \\
        = \bigl(x_{p, [k,l+p]}^* (\a) (X_{[k,l]} \otimes \id_{[l+1,l+p]})\bigr)
	  \otimes \id_{[l+p+1,\infty)} \epc
\end{multline}
where $x = b^*, c^*$.

Inserting the mode expansions (\ref{claurant}) and (\ref{blaurant}) for
the annihilation operators into the reduction relation (\ref{annihiliredu})
we obtain
\begin{equation}
     \sum_{p=0}^{n+m}
        \frac{x_{p,[k,l+m]} (\a) \bigl(X_{[k,l]} \otimes \id_{[l+1,l+m]}\bigr)}
	     {(\z^2 - 1)^p}
        = \sum_{p=0}^n
	     \frac{\bigl(x_{p,[k,l]} (\a) X_{[k,l]}\bigr) \otimes \id_{[l+1,m+1]}}
	          {(\z^2 - 1)^p} \epc
\end{equation}
where $x = b, c$, and $n = l - k + 1$ is the length of $X_{[k,l]}$.
Comparing coefficients we conclude that
\begin{multline}
     x_{p,[k,l+m]} (\a) \bigl(X_{[k,l]} \otimes \id_{[l+1,l+m]}\bigr) \\ =
        \begin{cases}
	   \bigl(x_{p,[k,l]} (\a) X_{[k,l]}\bigr) \otimes \id_{[l+1,l+m]} &
	   p = 1, \dots, n \\
	   0 & p = n + 1, \dots, m \epp
        \end{cases}
\end{multline}
Thus, any operator $X$ of length $\ell (X)$ is annihilated by
$x_p$ if $p > \ell (X)$. Since $m$ is arbitrary here we may
extend the action of the modes infinitely to the right,
\begin{multline} \label{annihireduright}
     x_{p,[k,\infty)} (\a) \bigl(X_{[k,l]} \otimes \id_{[l+1,\infty)}\bigr) \\ =
        \begin{cases}
	   \bigl(x_{p,[k,l]} (\a) X_{[k,l]}\bigr) \otimes \id_{[l+1,\infty)} &
	   p = 1, \dots, n \\
	   0 & p > n \epp
        \end{cases}
\end{multline}

Equations (\ref{treduright}), (\ref{crreduright}) and
(\ref{annihireduright}) comprise the reduction properties of
the modes for a reduction to the right. Before considering
the left reduction properties we would like to modify our formulation
in a way that allows us to define uniform mode expansions for
the action on operators that have no definite spin. For this
purpose note that the variable $\a$ is still at our disposal.
We may shift it in such a way that the spin dependence in
the mode expansions of the creation and annihilation operators
is moved from the spectral parameter to the Taylor and Laurent
coefficients. For the annihilation operators we obtain from
(\ref{claurant}), (\ref{blaurant})
\begin{equation} \label{sshiftmodeexp}
     \xv \bigl(\z, \a - s - s(\xv)\bigr) X_{[k,l]} =
        \z^{\a s(\xv)} \sum_{p=0}^n
	\frac{x_p \bigl(\a - s - s(\xv)\bigr) X_{[k,l]}}{(\z^2 - 1)^p}
\end{equation}
for $\xv = \bv, \cv$, $x = b, c$. Recall that $s(\bv) = -1$,
$s(\cv) = 1$. Similarly, the mode expansions for the creation
operators
can be written as
\begin{equation}
     \xv \bigl(\z, \a - s - s(\xv)\bigr) X_{[k,l]} =
        \z^{\a s(\xv)} \sum_{p=1}^\infty
	(\z^2 - 1)^{p-1} x_p \bigl(\a - s - s(\xv)\bigr) X_{[k,l]}
\end{equation}
for $\xv = \bv^*, \cv^*, \tv^*$. Here $s(\bv^*) = 1$, $s(\cv^*) = - 1$,
and $s(\tv^*) = 0$.

Inserting (\ref{tstarmodesredu}), (\ref{bstarmodesredu}) or
(\ref{cstarmodesredu}) we obtain
\begin{multline}
     \xv_{[k,l+m]} \bigl(\z, \a - s - s(\xv)\bigr)
        \bigl(X_{[k,l]} \otimes \id_{[l+1,l+m]}\bigr) \\ =
        \z^{\a s(\xv)} \sum_{p=1}^m
	(\z^2 - 1)^{p-1} \bigl(x_{p,[k,l+p]} \bigl(\a - s - s(\xv)\bigr)
	   (X_{[k,l]} \otimes \id_{[l+1,l+p]})\bigr) \otimes \id_{[l+p+1,l+m]} \\[-1ex]
        + \z^{\a s(\xv)} \times \CO \bigl( (\z^2 - 1)^m \bigr)
\end{multline}
for all $m \in {\mathbb N}$. From this equation we understand how
the creation operators can be extended infinitely to the right
by the inductive limit $m \rightarrow \infty$,
\begin{multline} \label{crearight}
     \xv_{[k,\infty)} \bigl(\z, \a - s - s(\xv)\bigr)
        \bigl(X_{[k,l]} \otimes \id_{[l+1,\infty)}\bigr) \\ =
        \z^{\a s(\xv)} \sum_{p=1}^\infty
	(\z^2 - 1)^{p-1} \bigl(x_{p,[k,l+p]} \bigl(\a - s - s(\xv)\bigr)
	   \\[-2ex] \times
	   (X_{[k,l]} \otimes \id_{[l+1,l+p]})\bigr) \otimes \id_{[l+p+1,\infty)}
\end{multline}
for $\xv = \bv^*, \cv^*, \tv^*$. Note that each term under the sum is of
finite length.
\begin{remark}
Comparing the last two equations we may also state that
\begin{multline}
     \xv_{[k,\infty)} \bigl(\z, \a - s - s(\xv)\bigr)
        \bigl(X_{[k,l]} \otimes \id_{[l+1,\infty)}\bigr) \\[1ex] =
        \bigl(\xv_{[k,l+m]} \bigl(\z, \a - s - s(\xv)\bigr)
	      (X_{[k,l]} \otimes \id_{[l+1,l+m]})\bigr) \otimes \id_{[m+1,\infty)} \\
        \mod \z^{\a s(\xv)} \times \CO \bigl( (\z^2 - 1)^m \bigr) \epc
\end{multline}
which is the way the inductive limit for the extension to the
right is introduced in \cite{BJMST08a}.
\end{remark}
For the extension of the mode expansion of the annihilation operators
to $+ \infty$ we use (\ref{anrredu}), (\ref{annihireduright}) and
(\ref{sshiftmodeexp}) to obtain
\begin{multline} \label{annihiright}
     \xv_{[k,\infty)} \bigl(\z, \a - s - s(\xv)\bigr)
        \bigl(X_{[k,l]} \otimes \id_{[l+1,\infty)}\bigr) \\ =
        \z^{\a s(\xv)} \sum_{p=1}^n
	\frac{\bigl(x_{p,[k,l]} \bigl(\a - s - s(\xv)\bigr) X_{[k,l]}\bigr)
	      \otimes \id_{[l+1,\infty)}}{(\z^2 - 1)^p}
\end{multline}
for $\xv = \bv, \cv$.

Equations (\ref{crearight}) and (\ref{annihiright}) describe
how the action of the creation and annihilation operators
can be extended infinitely to the right. We may use these equations
in (\ref{allleft}) in order to obtain operators acting on the entire
infinite chain. More precisely, these operators define maps
$\xv: {\cal W}_{\a - s, s} \rightarrow {\cal W}_{\a - s - s(\xv),
s + s(\xv)}$. Let $X^{(\a -s)} \in {\cal W}_{\a - s, s}$ with
non-trivial part $X_{[k,l]} \in {\cal H}_{[k,l]}$. Then
\begin{multline} \label{creafull}
     \xv \bigl(\z, \a - s - s(\xv)\bigr) X^{(\a - s)} \\ =
        \z^{\a s(\xv)} \sum_{p=1}^\infty (\z^2 - 1)^{p-1}
        q^{2(\a - s - s(\xv)) S_{(-\infty,k-1]}} \otimes
	\bigl(x_{p,[k,l+p]} \bigl(\a - s - s(\xv)\bigr)
	   \\[-2ex] \times
	   (X_{[k,l]} \otimes \id_{[l+1,l+p]})\bigr) \otimes \id_{[l+p+1,\infty)}
\end{multline}
for the creation operators $\xv = \tv^*, \bv^*, \cv^*$ and
\begin{multline} \label{annihifull}
     \xv \bigl(\z, \a - s - s(\xv)\bigr) X^{(\a - s)} \\ =
        \z^{\a s(\xv)} \sum_{p=1}^n
        q^{2(\a - s - s(\xv)) S_{(-\infty,k-1]}} \otimes
	\frac{\bigl(x_{p,[k,l]} \bigl(\a - s - s(\xv)\bigr) X_{[k,l]}\bigr)
	      \otimes \id_{[l+1,\infty)}}{(\z^2 - 1)^p}
\end{multline}
for the annihilation operators $\xv = \bv, \cv$. These
formulae show how the action of the finite operators $\xv_{[k,l]}$
introduced in the previous subsections can be naturally extended
to the infinite chain.

We define the modes $\xv_p: {\cal W}_{\a - s, s} \rightarrow
{\cal W}_{\a - s - s(\xv), s + s(\xv)}$ by their action
on $X^{(\a -s)} \in {\cal W}_{\a - s, s}$ with finite part
$X_{[k,l]}$. For the creation operators we set
\begin{multline} \label{creapremode}
     \xv_p \bigl(\a - s - s(\xv)\bigr) X^{(\a - s)} =
        q^{2(\a - s - s(\xv)) S_{(-\infty,k-1]}} \\ \otimes
	\bigl(x_{p,[k,l+p]} \bigl(\a - s - s(\xv)\bigr)
	   (X_{[k,l]} \otimes \id_{[l+1,l+p]})\bigr) \otimes \id_{[l+p+1,\infty)}
\end{multline}
where $\xv = \tv^*, \widetilde \bv^*, \widetilde \cv^*$ and
where the finite parts are $x_p = t_p^*, b_p^*, c_p^*$, whose
explicit form are given in (\ref{tstarmodes}), (\ref{bstarmodes}),
(\ref{cstarmodes}). For the annihilation operators we define
\begin{multline} \label{annihimode}
     \xv_p \bigl(\a - s - s(\xv)\bigr) X^{(\a - s)} =
        q^{2(\a - s - s(\xv)) S_{(-\infty,k-1]}} \\ \otimes
	\begin{cases}
	   \bigl(x_{p,[k,l]} \bigl(\a - s - s(\xv)\bigr) X_{[k,l]}\bigr)
	   \otimes \id_{[l+1,\infty)} & p = 1, \dots, n \\[1ex]
           0 & p > n \epc
	\end{cases}
\end{multline}
where $\xv = \bv, \cv$, and the finite parts $x_p = c_p, b_p$
have been defined in (\ref{cmodes}), (\ref{bmodes}). Altogether
we have obtained fully explicit expressions for the Fermionic
operators as maps ${\cal W}_{\a - s, s} \rightarrow
{\cal W}_{\a - s - s(\xv), s + s(\xv)}$. They are defined by
the mode expansions
\begin{subequations}
\label{modeexpblocks}
\begin{align}
     \tv^* (\z, \a - s) &
        = \sum_{p=1}^\infty (\z^2 - 1)^{p-1}
	  \tv^*_p (\a - s) \epc \\
     \xv^* \bigl(\z, \a - s - s(\xv^*)\bigr) &
        = \z^{\a s(\xv^*)} \sum_{p=1}^\infty (\z^2 - 1)^{p-1} \,
	  \widetilde \xv_p^* \bigl(\a - s - s(\xv)\bigr) \epc \\
     \xv \bigl(\z, \a - s - s(\xv)\bigr) &
        = \z^{\a s(\xv)} \sum_{p=1}^\infty (\z^2 - 1)^{-p} \,
	  \xv_p \bigl(\a - s - s(\xv)\bigr) \epc
\end{align}
\end{subequations}
where $\xv = \bv, \cv$.

We may consider these operators as `${\cal W}_{\a - s, s}
\rightarrow {\cal W}_{\a - s - s(\xv), s + s(\xv)}$ blocks'
of operators $\xv (\z) \in \End {\cal W}^{(\a)}$ (cf.\
(\ref{defwalpha})). If $X \in {\cal W}^{(\a)}$ with
$\ell (X) = n$, then
\begin{equation}
     X = \sum_{s = - n}^n X^{(\a - s)} \epc
\end{equation}
where $X^{(\a - s)} \in {\cal W}_{\a -s, s}$, and
\begin{equation}
     \xv (\z) X = \sum_{s = - n}^n \xv(\z, \a - s - s(\xv)) X^{(\a - s)} \epp
\end{equation}
In a similar way we can define the action of the modes in
(\ref{modeexpblocks}) on ${\cal W}^{(\a)}$,
\begin{equation}
     \xv_p X = \sum_{s = - n}^n \xv_p (\a - s - s(\xv)) X^{(\a - s)} \epp
\end{equation}
The generating functions $\xv (\z) = \tv^* (\z), \bv^* (\z),
\cv^*(\z), \bv (\z), \cv(\z)$ hence have the mode expansions
\begin{subequations}
\label{premodes}
\begin{align}
     \tv^* (\z) & = \sum_{p=1}^\infty (\z^2 - 1)^{p-1}\, \tv^*_p \epc \\
     \xv^* (\z) & = \z^{\a s(\xv^*)} \sum_{p=1}^\infty (\z^2 - 1)^{p-1}\,
                    \widetilde \xv^*_p \epc \label{creaprefinal} \\
     \xv (\z) & = \z^{\a s(\xv)} \sum_{p=0}^\infty \frac{\xv_p}{(\z^2 - 1)^p} \epc
     \label{annihifinal}
\end{align}
\end{subequations}
where $\xv = \bv, \cv$. Recall that the action of the modes
on the right hand side on ${\cal W}_{\a - s, s}$ is defined
by (\ref{creapremode}), (\ref{annihimode}).

We use the geometric series
\begin{equation}
     \z^2 \sum_{k=0}^\infty (-1)^k (\z^2 - 1)^k = 1
\end{equation}
in order to resum the modes in the mode expansion
(\ref{creaprefinal}) for $\bv^* (\z)$ and $\cv^* (\z)$. Then
\begin{align} \label{creafinal}
     \xv^* (\z) & = \z^{\a s(\xv^*) + 2}
                    \sum_{p=1}^\infty \sum_{k=0}^\infty (\z^2 - 1)^{p + k - 1}
		       (-1)^k\, \widetilde \xv^*_p \notag \\[1ex]
                & = \z^{\a s(\xv^*) + 2}
                    \sum_{p=1}^\infty (\z^2 - 1)^{p - 1}\, \xv^*_p \epc
\end{align}
where
\begin{equation}
     \xv^*_p = \sum_{k=1}^p (-1)^{p-k}\, \widetilde \xv_k^*
\end{equation}
for $p \in {\mathbb N}$ and $\xv^* = \bv^*, \cv^*$. As opposed to
the modes $\widetilde \xv^*_p$ the modes $\xv^*_p$ satisfy canonical
anticommutation relations with the modes of the corresponding
annihilation operators (see below).

To sum up, we have presented a fully explicit description of the
operators in the Fermionic basis which is suitable for
its implementation in a computer algebra program.

\subsection{Remarks on the commutation relations}
The commutation relations of the operators acting on finite
chains are discussed in section 4 of \cite{BJMST08a} together
with the commutation relations of the operators acting on
the space of quasi-local operators on the infinite chain.
The proofs are at the same time very technical and rather
sketchy. They are certainly the most challenging part of the
theory. Here we extract the information about the operators defined
on the finite chain which is needed if we wish to verify their
commutation relations in special cases on a computer.

Section 4 of \cite{BJMST08a} provides the commutation relations
in a somewhat implicit form involving what the authors call
$q$-exact forms. The commutativity of the Fermionic annihilation
operators with $\tv^*$, for instance, is expressed in equation
(4.3) of \cite{BJMST08a} as
\begin{multline} \label{ktcom}
    \kv (\x, \a) \bigl((\tv^* (\z, \a + 1) X_{[k,m]})
                        \otimes \id_{[m+1,l]}\bigr) -
       \tv^* (\z, \a) \bigl((\kv (\x, \a) X_{[k,m]})
       \otimes \id_{[m+1,l]}\bigr) \\[1ex]
       \simeq_\x 0 \mod (\z^2 - 1)^{l-m}
\end{multline}
for $k \le m < l$. Here `$\simeq_\x$' means equality up to a
$q$-exact form in $\x$, a notion that is explained in section~2.6
of \cite{BJMST08a}. It means that the right hand side of
(\ref{ktcom}) can be written as $\D_\x \, \x^{\a + s} m(\x^2, \a)$,
where $m$ is rational in $\x^2$ with a pole at most at $\x^2 = 1$ and
where $s$ is the spin of $X_{[k,m]}$. It follows that
\begin{align}
    & \bigl(\kv (q \x, \a) + \kv (q^{-1} \x, \a)\bigr)
    \bigl((\tv^* (\z, \a + 1) X_{[k,m]}) \otimes \id_{[m+1,l]}\bigr)
       \notag \\ & \mspace{162.mu}
       - \tv^* (\z, \a)
    \bigl(\bigl((\kv (q \x, \a) + \kv (q^{- 1} \x, \a)) X_{[k,m]} \bigr)
       \otimes \id_{[m+1,l]}\bigr)
       \notag \\[1ex] & \qd =
    \D_\x (q \x)^{\a + s} m \bigl((q \x)^2, \a\bigr)
       + \D_\x (q^{-1} \x)^{\a +s} m \bigl((q^{-1} \x)^2, \a\bigr)
       \mod (\z^2 - 1)^{l-m} \notag \\[1ex] & \qd
    = (q^2 \x)^{\a + s} m(q^4 \x^2, \a) - (q^{-2} \x)^{\a +s} m(q^{-4} \x^2, \a) 
       \mod (\z^2 - 1)^{l-m} \epp
\end{align}
If we multiply by $\ps (\z_1/\x, \a + s +1)/(4 \p \i \x^2)$, then both
sides are rational in $\x^2$ and the right hand side is obviously
regular at $\x^2 = 1$. Integrating over $\x^2$ on a small circle
around $\x^2 = 1$ and using the definition of $\cv$, Eq.~ (\ref{defcint}),
we obtain
\begin{multline} \label{ctcom}
    \cv (\x, \a) \bigl((\tv^* (\z, \a + 1) X_{[k,m]})
                        \otimes \id_{[m+1,l]}\bigr) -
       \tv^* (\z, \a)
       \bigl((\cv (\x, \a) X_{[k,m]}) \otimes \id_{[m+1,l]}\bigr) \\[1ex]
       = 0 \mod (\z^2 - 1)^{l-m} \epp
\end{multline}
In a similar way we can obtain all other commutation relations
of the finite generating functions from the relations in section~4
of \cite{BJMST08a}. If we shift $\a \rightarrow \a - s - 1$ and
perform the inductive limit described in the previous subsection
on equation (\ref{ctcom}), we obtain the commutation relation
$[\cv(\x), \tv^*(\z)] = 0$. The ideology for the derivation of
the other commutation relations (\ref{commutism}),
(\ref{anticommutism}) is similar. In any case, the hard part of
the proof is to derive relations like (\ref{ktcom}).

We will not touch this subject any deeper. Yet, we would like
to comment on how the commutation relations for the modes follow
from those for the generating functions. They are obtained
by inserting the mode expansions into the commutation relations
for the generating functions. Whenever two generating functions
commute or anticommute the same is trivially true also for
the corresponding modes. The only cases which need extra attention
are the anticommutators $[\bv(\z_1),\bv^\ast(\z_2)]_+$ and $[\cv(\z_1),
\cv^\ast(\z_2)]_+$ which can be uniformly written as
\begin{equation}
     [\xv(\z_1),\xv^* (\z_2)]_+ = \ps \bigl(\z_1/\z_2, \a s(\xv)\bigr)
\end{equation}
for $\xv = \bv, \cv$. Inserting the mode expansions (\ref{annihifinal}),
(\ref{creafinal}) on the left hand side we obtain an equation which
is equivalent to
\begin{equation}
     \sum_{k=0}^\infty \sum_{p=1}^\infty [\xv_k, \xv_p^*]_+
        \frac{(\z_2^2 - 1)^{p-1}}{(\z_1^2 - 1)^k} = 
	\frac{\z_1^2 + \z_2^2}{2\z_2^2(\z_1^2 - \z_2^2)} \epp
\end{equation}
Using (\ref{psidetox}) and comparing coefficients we obtain the
anticommutation relations
\begin{subequations}
\begin{align}
     [\xv_0, \xv_p^*]_+ & = \2 (-1)^{p-1} \epc && p \in {\mathbb N} \epc \\
     [\xv_k, \xv_p^*]_+ & = \de_{k, p} \epc && k, p \in {\mathbb N} \epp
\end{align}
\end{subequations}
Here the first set of relations is consistent with the definitions
of the the modes $\xv_0$ (see (\ref{cnull}), (\ref{bmodes})) while 
the second set of relations is just the set of canonical
anticommutation relations for Fermions. At this point it is
also clear that the coefficients $\widetilde \xv_p^*
= \xv_p^* + \xv_{p-1}^*$ do not satisfy the canonical
anticommutation relations.

\subsection{Remarks on the implementation on the computer}
\newcommand{\Lbb}{\mathbb{L}}
As explained in the previous subsections we are able to represent
operators acting on the infinite chain in terms of finite matrices.
Hence, it is possible to construct such operators explicitly on
the computer. To do so a system is needed in which symbolic expressions,
including non-commutative symbols, can be manipulated. For this
purpose we have used FORM \cite{Vermaseren00}. While there exist packages
for e.g.\ Mathematica for this purpose, our experience is, that
full blown computer algebra systems like Mathematica are not efficient
enough in terms of memory management. In comparison, FORM is a rather
primitive language with less features, which allows for a much simpler
internal representation of expressions, resulting in a more efficient
memory management.

As pointed out before, no closed formula is known for expanding
products of ultra-local operators (e.g. $\sigma_1^z\sigma_n^z$)
in the Fermionic basis. For this reason as well as for a few others
\cite{Kleinemuehl20}, the direct computation of correlation functions
by means of the JMS theorem is rather inefficient. Instead we used
the so-called exponential form which means that we only need to
construct the modes of the annihilation operators $\bv,\cv$.

In order to obtain these, in first place, the parental operator $\kv$
needs to be constructed. As discussed before, the construction can
be done for the case of a finite chain.
For our program we changed formula (\ref{eq:k}) in order to express
$\kv_{\mathrm{skal},[k,l]}$ in terms of the fused $L$-matrices
$L_{\{a,A\},j}$ introduced in \cite{BJMST08a}:
\begin{equation}
  L_{\{a,A\},j}
  =
  F_{a,A}^{-1} L_{a,j} L_{A,j} F_{a,A}
  \ ,\qquad
  F_{a,A} = 1 - \av_A \sigma_a^+ \epc
\end{equation}
which can be written explicitly as
\begin{equation}
  L_{\{a,A\},j}(\z)
  =
  \begin{pmatrix}
    1   &  0 \\
    \frac{\gamma(\zeta)}{\beta(\zeta)}\sigma_j^+  &  1
  \end{pmatrix}_a
  \begin{pmatrix}
    L_{A,j}(q\zeta)q^{-\sigma_j^z/2}    &  0 \\
    0                                   &  L_{A,j}(q^{-1}\zeta)q^{\sigma_j^z/2}
  \end{pmatrix}_a\ .
\end{equation}
It is possible to express $\kv_{\mathrm{skal},[k,l]}$ in terms of the adjoint action of this $L$-matrix,
\begin{equation}
  \Lbb_{\{a,A\},j} \id_a \otimes \id_A \otimes X_{[k,l]}
  =
  L_{\{a,A\},j}
  \id_a \otimes \id_A \otimes X_{[k,l]}
  L_{\{a,A\},j}^{-1} \epc
\end{equation}
since $[F_{a,A}, \sigma_a^+] = [F_{a,A}, q^{\a(\sigma_a^z + 2D_A)}] = 0$.

Our program then uses
\begin{multline}
  \kv_{\mathrm{skal},[k,l]}(\z,\a) X_{[k,l]} \\ =
  \tr_{a,A}\lb\{ \sigma_a^+
  \Lbb_{\{a,A\},l}(\z) \cdots \Lbb_{\{a,A\},k}(\z)
  \lb( \z^{-2s-1} y^{\sigma_a^z + 2D_A} q^{-2S_{[k,l]}} X_{[k,l]} \rb)
  \rb\} \epc
\end{multline}
where again the operator $X_{[k,l]}$ is of spin $s$. We also set
$y = q^\a$ in our program because it is only a single symbol.
Using the fused operators $L_{\{a,A\},j}$ instead of the simple
$L$-operators effectively means an early simplification of the
building blocks of $\kv$.

It is then convenient to compute the action of
$\kv_{\mathrm{skal},[k,l]}(\z,\a)$ on the canonical basis of
$\End {\cal H}_{[k,l]}$ constructed with the single-entry matrices
$e_{j\a}^{\phantom{j}\beta}$. This will also allow for an easy
parallelization later. For each element of the basis the
`innermost' part
\begin{equation}
  \z^{-2s-1} y^{\sigma_a^z + 2D_A} q^{-2S_{[k,l]}} X_{[k,l]} 
\end{equation}
is constructed first. The spin-$\frac{1}{2}$ auxiliary space is
explicitly used, whereas $q^{2\a D_A}$ is represented by a single
non-commuting symbol. Then, in a loop, each step applies a single
fused $L$ matrix $\Lbb_{\{a,A\},j}(\z)$ after which all symbols
are commuted and sorted. When the loop is finished, the elements
$e_{a+}^{\phantom{a}+}$, $e_{a-}^{\phantom{a}-}$ and
$e_{a+}^{\phantom{a}-}$ can be discarded, because of the operator
$\sigma_a^+$ and the trace $\tr_a$. The remaining trace $\tr_A$
can then be taken by discarding all terms which are not `balanced'
in $\av_A^\ast$ and $\av_A$ and replacing $y^{2D_A} q^{m D_A}
\rightarrow \frac{1}{1 - y^2 q^m}$.

Since the intermediate expressions can become very large (a few
gigabytes for $n=5$), it is crucial to carefully choose which
simplification is done at which point. Too many simplifications
can slow down the calculations, but, on the other hand, too few
can increase the memory usage dramatically. There is no definite
rule in this regard, and the appropriate places in the program,
where simplifications are most efficient, need to be determined
for each calculation individually. It is also helpful to use simple
symbols wherever possible, e.g.\ representing the functions
$\beta(\z), \gamma(\z)$ by single symbols rather than rational
expressions. Such symbols should then be expanded at the latest
possible stage.

After constructing the operator $\kv_{\mathrm{skal},[k,l]}$, the
Laurent-coefficients are obtained. This is done slightly differently
from (\ref{eq:laurent_rho}) and (\ref{eq:laurent_kappa}). The reason
is that using (\ref{eq:laurent_rho}) and (\ref{eq:laurent_kappa})
directly would trigger the same intermediate calculations to be performed 
multiple times. Instead, $\kv_{\mathrm{skal},[k,l]}(\z,\a)$ is
loaded and a loop counts down $j = n,\dots,1$. For each $j$ we run 
over $\eps=-1,0,1$ and set
\begin{subequations}
\begin{align}
	\rho_{j,[k,l]}^{(\eps)}(\a) & =
 	\lb.(\z^2 - q^{2\eps})^j\ \kv^{(j,\eps)}_{\mathrm{skal},[k,l]}(\z,\a)
 	\rb|_{\z^2 = q^{2\eps}} \epc \\[1ex]
	\kappa_{j,[k,l]}(\a) & =
 	\lb.\z^{2j}\ \kv^{(j,2)}_{\mathrm{skal},[k,l]}(\z,\a)
 	\rb|_{\z^2=0} \epc
\end{align}
\end{subequations}
where
\begin{equation}
	\kv^{(j,\epsilon)}_{\mathrm{skal},[k,l]}(\z,\a) =
	\kv_{\mathrm{skal},[k,l]}(\z,\a)
	-
	\sum_{m=j}^n \sum_{\delta = -1}^{\epsilon-1}
	\frac{\rho_{m,[k,l]}^{(\delta)}(\a)}{(\z^2-q^{2\delta})^m}
	-
	\sum_{m=j+1}^n \frac{\kappa_{m,[k,l]}(\a)}{\z^{2m}} \epp
\end{equation}

The $\kv^{(j,\epsilon)}_{\mathrm{skal},[k,l]}(\z,\a)$ can then be
`accumulated' during the loop, instead of building them anew in 
every step, and therefore become smaller in every iteration of the
loop. During traversal of the loop $\kappa_0(\a)$ is accumulated
as well, which contains all $\rho_j^{(\a)}$. Since the $\rho_j^{(0)}$
are only needed for $\kappa_0(\a)$, they are discarded as soon as
they have entered $\kappa_0(\a)$. Each time one of the Laurent
coefficients is completed, a sorting is done in order to keep
$\kv^{(j,\epsilon)}_{\mathrm{skal},[k,l]}(\z,\a)$ as small as possible.
The modes of the annihilation operators are then easily constructed
according to (\ref{tcmodes}) and (\ref{bmodes}).

In order to verify the correctness of the obtained operators we
utilize the \mbox{(anti-)}com\-mutation relations. The modes of
the transfer matrix $\tv^\ast$ have to commute amongst themselves
as well as with all modes of the Fermionic operators. The modes
of $\bv, \bv^\ast$ and $\cv, \cv^\ast$ form two families of
Fermions and obey the canonical relations as noted before.
For the sake of brevity we do not go into detail regarding the
construction of the creation operators, since they are not
directly needed when using the exponential form. We did, however,
construct all of them explicitly, so we were able to verify all
\mbox{(anti-)}com\-mutation relations that are sensibly defined
on an interval of a given length~$n$. Additionally, for the
annihilation operators, the annihilation relations (\ref{annihipropmodes})
and reduction relations (\ref{leftredubyone}), (\ref{annihiliredu}),
(\ref{bstarmodesredu}), (\ref{cstarmodesredu}) were verified.
In our experience most of these relations depend on the correctness
of the involved operators in a very sensitive manner. Typically,
even small errors in the program led to objects which obey none
of the tested relations. So, if the operators constructed obey
all above relations, it is a strong indication of their correctness.
\section{Correlation functions}
\label{sec:cf}
In this section we use the exponential form (\ref{expkappa})
in order to study short-range correlation functions of 
spin-reversal invariant operators at finite temperature.
In Appendix~\ref{app:jms_to_expform} we show how
(\ref{expkappanull}) follows from the JMS theorem and argue
that it continues to hold for $\k \ne 0$ if we restrict the
action of $Z^\k$ to spin-reversal invariant operators.

\subsection{The exponential form}
Inserting the mode expansions (\ref{annihimodes}) for the
annihilation operators into (\ref{defomop}) and using Cauchy's
theorem we obtain a mode expansion of the operator $\Om$,
\begin{equation}
     \Om = \sum_{k,p=1}^\infty \frac{\om_{k-1, p-1} (\a)\, \bv_k \cv_p}
                                    {(k-1)! (p-1)!} \epc
\end{equation}
where
\begin{equation} \label{omderivs}
     \om_{k, p}
        = \6_{\z_1^2}^k \6_{\z_2^2}^p \; \z_1^{-2} \z_2^{-2}
	  \bigl(\z_1/\z_2\bigr)^{- \a}
	  \bigl(\om_0 (\z_1/\z_2; \a) - \om(\z_1, \z_2; \a)\bigr)
	  \Bigr|_{\z_1^2 = \z_2^2 = 1} \epp
\end{equation}

Recall that we denoted by ${\cal W}_{\a, 0, [1,n]} \subset
{\cal W}_{\a, 0}$ the space of quasi-local operators of spin~$0$
with tail $\a$ and with support $\supp X^{(\a)} \subset [1,n]$.
Due to the annihilation property (\ref{annihipropmodes})
of the modes, the restriction of $\Om$ to this subspace is
a finite sum,
\begin{equation} \label{ommodes}
     \Om\bigr|_{{\cal W}_{\a, 0, [1,n]}} =
        \sum_{k,p=1}^n \frac{\om_{k-1, p-1} (\a)\, \bv_k \cv_p}
	                    {(k-1)! (p-1)!} \epp
\end{equation}
In \cite{BJMS09a} the authors constructed a basis of
${\cal W}_{\a, 0, [1,n]}$ as a submodule of the Fermionic basis.
This submodule is generated by the action of monomials of the form
\begin{equation}
     \tv^*_{p_1} \dots \tv^*_{p_j} \bv^*_{q_1^+} \dots \bv^*_{q_k^+}
        \cv^*_{q_k^-} \dots \cv^*_{q_1^-}
\end{equation}
onto the vacuum $q^{2 \a S(0)}$, where
\begin{equation}
     j + 2k \le n \epc \qd p_l, q_m^+, q_n^- \le n \epp
\end{equation}
In particular, the number $k$ of creation operators $\bv^*_{q_m^+}$,
$\cv^*_{q_n^-}$ of Fermions is restricted by
\begin{equation}
     k \le \Bigl[\frac{n}{2}\Bigr] \epc
\end{equation}
where the bracket denotes the integer part. It follows that
\begin{equation}
     \Bigl(\Om\bigr|_{{\cal W}_{\a, 0, [1,n]}}\Bigr)^{[\frac{n}{2}] + 1} = 0 \epc
\end{equation}
and therefore
\begin{equation} \label{expnil}
     \re^\Om\bigr|_{{\cal W}_{\a, 0, [1,n]}}
        = \sum_{k=0}^{[\frac{n}{2}]} \frac{(\Om|_{{\cal W}_{\a, 0, [1,n]}})^k}{k!} \epp
\end{equation}
Using (\ref{cmodes}), (\ref{bmodes}), (\ref{annihimode}) and
(\ref{ommodes}), this can be realized by means of a finite sum over
finite products of finite matrices and is the formula used below
for the calculation of short-range correlation functions.

In order to explicitly calculate correlation functions on the
computer we use (\ref{expkappa}), (\ref{ommodes}) and (\ref{expnil}).
At this point we have shown that all sums involved are finite
and all operators can be represented in terms of finite matrices.
We do not evaluate the exponential given in (\ref{expnil}) directly
for two reasons. On the one hand we argue that our method holds
for operators that are invariant under spin reversal, meaning
that it is only valid on a subspace of ${\cal W}_{\a, 0, [1,n]}$.
On the other hand the involved expressions tend to grow rather big,
making it important to save as much memory as possible. For these
reasons we only calculate the action of the exponential on a given
operator $X^{(\a)} \in {\cal W}_{\a, 0, [1,n]}$.

We apply $\Om|_{{\cal W}_{\a, 0, [1,n]}}$ repeatedly to $X^{(\a)}$,
filtering out every term that becomes zero due to the nilpotence
of the modes before inserting explicit matrices and simplifying
as much as possible before the next step. During this process the
sum (\ref{expnil}) can be accumulated. The treatment of the
functions $\omega$ and $\hat\omega$ is explained in the next
section and in Appendix~\ref{app:multtoadd}.

This process then results in expressions rational in $q$ that
contain $\hat\omega$ and its derivatives as symbols, as shown in
\cite{BDGKSW08} for $n=1\dots4$ and in Appendix~\ref{app:n5_explicit}
for $n=5$.

\subsection{Finite temperature short-range correlation functions}
As was explained in the introduction, the Fermionic basis
approach applies to very general situations. It holds for any
realization of the functional $Z^\k$ which, in turns, can be
used to realize a huge class of reduced density matrices
including the cases of the canonical ensemble and of generalized
Gibbs ensembles. The characteristics of any functional $Z^\k$
enter the formalism only through two functions $\r$ and
$\om$. These functions were called the `physical part'
of the problem in \cite{BoGo09} as they entirely fix `the
experimental conditions' under which the correlation
functions under consideration are determined. For the
actual computation of short-range correlation functions
what is still needed is an efficient description of the
physical part.

Such a description, valid for the case of finite-temperature
correlations in the infinite chain, was obtained in \cite{BoGo09}.
The core part of this description is a non-linear integral
equation for an auxiliary function $\fa$ that had occurred
before in the derivation of an efficient thermodynamics of
the XXZ chain \cite{Kluemper93} and in the derivation of
multiple-integral representations for static finite-temperature
correlation functions of the same model \cite{GKS04a}.
Let us define the `bare energy function'
\begin{equation}
     \re(\la) = \cth(\la) - \cth(\la - \i \g)
\end{equation}
and the kernel function
\begin{equation}
     K_\a (\la) = q^{- \a} \cth (\la + \i \g) - q^\a \cth (\la - \i \g) \epc
\end{equation}
where $- \i \g = \ln q$. Using these functions the nonlinear integral
equation for $\fa$ takes the form
\begin{equation} \label{nlietemhom}
     \ln (\fa (\la, \k)) = 2 \i \g \k + \frac{2\i J\sin(\g) \re (\la)}{T}
        - \int_{\cal C} \frac{\rd \m}{2 \p \i}
		        K_0 (\la - \m) \ln (1 + \fa (\m, \k )) \epp
\end{equation}
Here the `Trotter limit', corresponding to infinitely many
horizontal lines in the definition of $Z^\k$, is already taken
(for more details see \cite{BoGo09}). The magnetic field $h > 0$
appearing in the Hamiltonian can be taken into account by setting
\begin{equation}
     \k = \frac{\i h}{2 \g T} \epp
\end{equation}
The precise definition of the integration contour depends on the
parameter regime. For $\g \in (0,\p/2)$ (implying that $0 < \D < 1$)
and $h > 0$, for instance, we may take the contour sketched in
Figure~\ref{fig:incontcrit}.
\begin{figure}
\centering
\includegraphics[width=.60\textwidth,angle=0,clip=true]{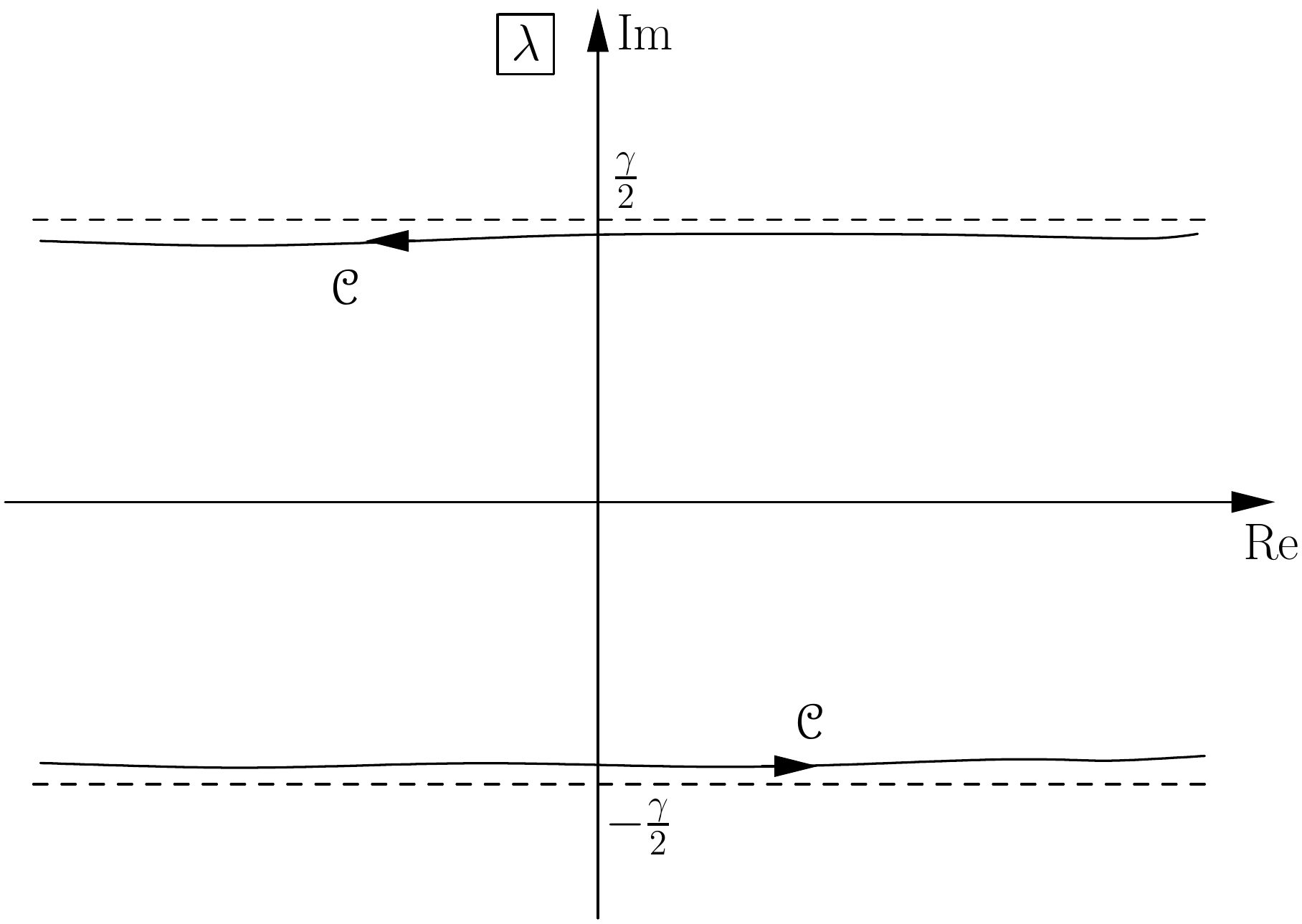}
\caption{\label{fig:incontcrit} For $0 < \D < 1$, $h > 0$ the canonical
contour ${\cal C}$ surrounds the real axis in a counterclockwise
manner inside the strip $- \frac{\g}{2} < \Im \la < \frac{\g}{2}$.} 
\end{figure}

The function $\r$, defined as an eigenvalue ratio in (\ref{defrho}),
has an integral representation involving the auxiliary function
and the bare energy,
\begin{equation} \label{rhoint}
     \r(\z) = \hat \r (\la) = q^\a \exp \biggl\{
                  \int_{\cal C} \frac{\rd \m}{2 \p \i} \: \re (\m - \la)
		  \ln \biggl( \frac{1 + \fa (\m, \k + \a)}
		                   {1 + \fa (\m, \k )} \biggr) \biggr\} \epc
\end{equation}
where $\la = \ln \z$. For the definition of $\om$ in the
finite temperature case we first of all need to introduce a function
$G$ which is the unique solution of the linear integral equation
\begin{equation} \label{newg}
     G(\la, \n) = q^{-\a} \cth(\la - \n + \i \g) - \hat \r (\n) \cth (\la - \n) 
                  + \int_{\cal C} \frac{\rd \m}{2\p \i}
		       \frac{K_\a (\la - \m) G(\m, \n)}
		            {\hat \r (\m) (1 + \fa(\m,\k))} \epp
\end{equation}
For $\la_j = \ln \z_j$, $j = 1, 2$, we set
\begin{equation} \label{Psi}
     \Ps (\z_1,\z_2; \a) = 
	\int_{\cal C} \frac{\rd \m}{2\p \i}
	   \frac{\bigl(q^\a \cth(\m - \la_1 + \i \g)
	               - \hat \r(\la_1) \cth(\m - \la_1) \bigr) G(\m, \la_2)}
		{\hat \r (\m)(1 + \fa(\m,\k))} \epp
\end{equation}
Then
\begin{multline} \label{ompsi}
     \om(\z_1, \z_2; \a) \\ = 2 (\z_1/\z_2)^\a \Ps(\z_1, \z_2; \a) - \D \ps(\z_1/\z_2, \a)
                       + 2 \bigl( \r(\z_1) - \r(\z_2) \bigr) \ps(\z_1/\z_2, \a) \epp
\end{multline}
This function has to be used in (\ref{omderivs}) for the actual
computation of the short-range correlation functions.

As we have to take the limit $\a \rightarrow 0$ in (\ref{expkappa})
and as the Fermi operators have first order poles in $\a$
by construction, it suffices to consider the function $\om$ up
to first order in~$\a$. Let 
\begin{equation}
     \hat \om (\la_1, \la_2; \a)
	= \bigl(\z_1/\z_2\bigr)^{- \a}
	  \bigl(\om_0 (\z_1/\z_2; \a) - \om(\z_1, \z_2; \a)\bigr) \epc
\end{equation}
where $\la_j = \ln \z_j$, $j = 1, 2$. The functions that occur in
the description of the physical correlation functions in the limit
$\a \rightarrow 0$ can be chosen as
\begin{subequations}
\label{omsnumerics}
\begin{align}
     \hat \om (\la_1, \la_2) & = \hat \om (\la_1, \la_2; 0) \epc \\
     \hat \om' (\la_1, \la_2)
        & = \2 \6_\a \bigl(\hat \om (\la_1, \la_2; \a)
	                   - \hat \om (\la_2, \la_1; \a)\bigr)\bigr|_{\a = 0} \epp
                             \label{antiomeprime}
\end{align}
\end{subequations}
Only the antisymmetric combination (\ref{antiomeprime}) remains
in the $\a \rightarrow 0$ limit of equation (\ref{ommodes}).
The choice (\ref{omsnumerics}) is convenient for the numerical
evaluation of the short-range correlation functions. As has been
shown in Appendix~B of \cite{BoGo09} these functions coincide with
the functions $\om (\la_1, \la_2)$ and $\om' (\la_1, \la_2)$
described in Section~4 of \cite{BDGKSW08}.

We have calculated the Fermionic basis representation of the
two-point functions $\<\s_1^z \s_n^z\>$ and $\<\s_1^x \s_n^x\>$
for $n = 2, 3, 4, 5$ from (\ref{expkappa}), (\ref{expnil}).
At the last stage we have used Appendix~\ref{app:multtoadd}
to replace the $\om_{k,p}$ by derivatives of $\hat \om$ and
$\hat \om'$. The results for $n = 2, 3, 4$ are the same as
previously obtained in the inhomogeneous case \cite{BGKS07,%
BDGKSW08}. The formulae for $n = 5$ are new. They are shown
in Appendix~\ref{app:n5_explicit}. We have then used the
representations of $\hat \om (\la_1,\la_2)$ and
$\hat \om' (\la_1, \la_2)$ derived in \cite{BDGKSW08} in order
to evaluate the two-point functions numerically.

\begin{figure}
\centering
\includegraphics[width=.98\textwidth,clip=true]{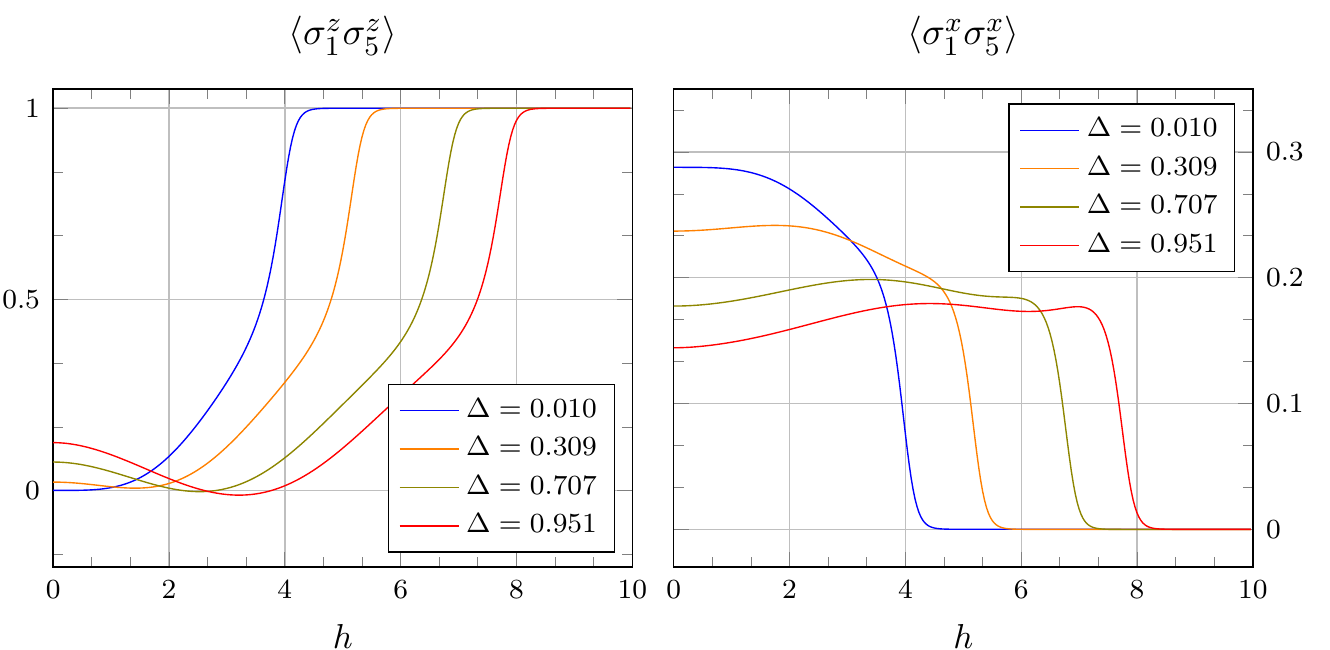}
\caption{\label{fig:n5_hvar_Deltapar_Tfix}
Fourth-neighbour correlation functions as functions of the magnetic
field for various values of $\D$ in the massless regime and $T/J = 0.1$.}
\end{figure}
The graphs show a rich, non-monotonous behaviour of the
correlation functions, reflecting the interplay of temperature,
external magnetic field and the relative strength of the Ising
and exchange interactions. Fig.~\ref{fig:n5_hvar_Deltapar_Tfix}
shows how the fourth-neighbour two-point functions depend on
the magnetic field at a relatively low temperature of $T/J = 0.1$
and for various values of the anisotropy parameter $\D$. The
longitudinal correlation functions saturate for $h$ above
the upper critical field $h_u = 4J(\D + 1)$, where the transverse
correlation functions show a complementary behaviour and vanish.
At $h = 0$ the longitudinal correlation functions show larger
positive correlations for larger $\D$, corresponding to an
increased Ising interaction, which is the intuitively expected
behavior. With increasing magnetic field the correlation functions
first diminish before they start growing again in opposite
order, such that the correlations are largest for the smallest
Ising interaction. This is perhaps somewhat counterintuitive, but
is at least in accordance with the fact that smaller $\D$ corresponds
to a smaller saturation field.

\begin{figure}
\centering
\includegraphics[width=.98\textwidth,clip=true]{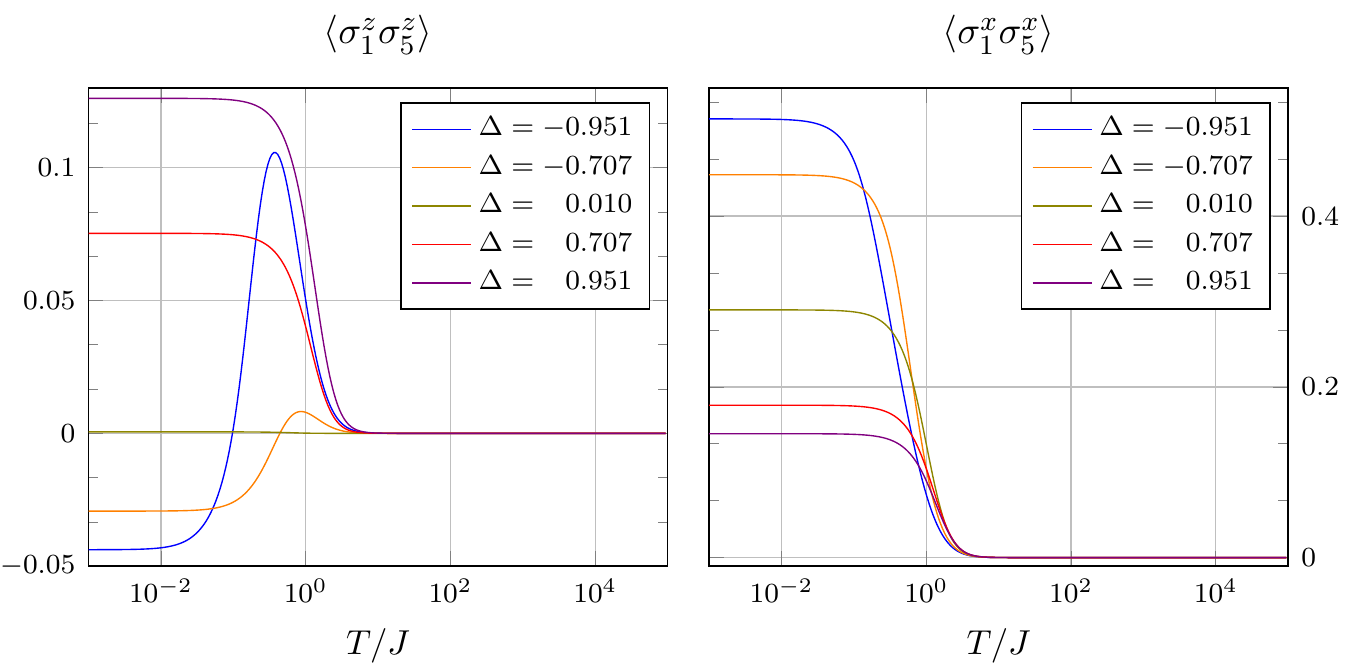}
\caption{\label{fig:n5_Tvar_Deltapar_hfix}
Fourth-neighbour correlation functions as functions of $T/J$
for various values of $\D$ in the massless regime and $h = 0$.}
\end{figure}
Fig.~\ref{fig:n5_Tvar_Deltapar_hfix} shows the temperature dependence
of the two-point functions at $h= 0$ and for various values of
$\D$ between $-1$ and $1$. For negative values of $\D$ we observe
a characteristic sign change of the longitudinal correlation
functions as the temperature increases. This sign change was
first discovered in a numerical study \cite{FaMc99}, where it
was interpreted as a `quantum to classical crossover'.
Fig.~\ref{fig:n5_Tvar_hpar_Deltafix} shows again the
temperature dependence of the fourth-neighbour two-point
function, this time at a fixed value $\D = 0.707$ of the
anisotropy for various magnetic fields. We see that the
correlations may change monotonously or non-monotonously
depending on the value of the external field.
\begin{figure}
\centering
\includegraphics[width=.98\textwidth,clip=true]{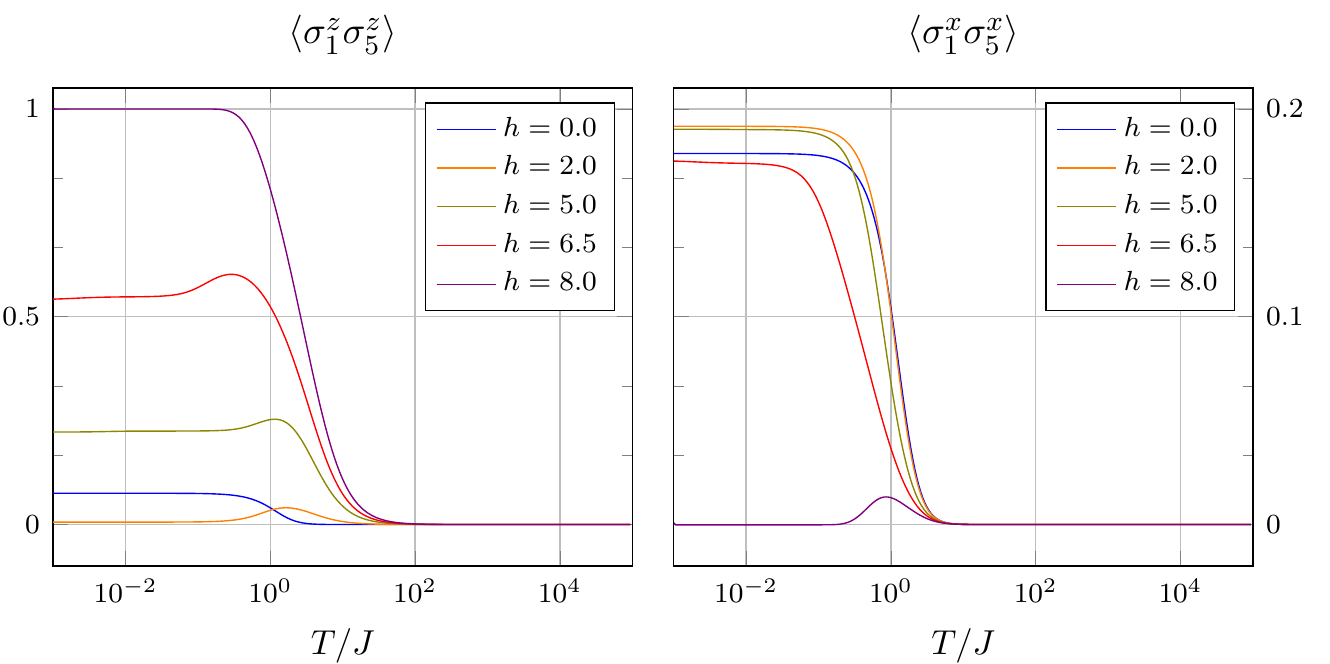}
\caption{\label{fig:n5_Tvar_hpar_Deltafix}
Fourth-neighbour correlation functions as functions of $T/J$
for various magnetic fields and $\D = 0.707$.}
\end{figure}

For the sake of completeness and also for comparison we have
included plots of the two-point functions of shorter range, 
$n = 2, 3, 4$, in Figs.~\ref{fig:n24_Tvar_Deltapar_hfix}-%
\ref{fig:n24_hvar_Deltapar_Tfix}. We have shown these plots
before in \cite{BDGKSW08} with the same choice of parameters.
But when we recomputed them for the present work we noticed that
the data for anisotropy parameters close to $\D = 1$ did not
have the accuracy claimed in that earlier work. This concerns
mostly the plots for $n = 4$ and $\D = 0.995$ for which the
numerical error was of an order of magnitude that could be
recognized with the naked eye. The reason for the numerical
error is that the individual terms in the Fermionic basis
representation of the correlation functions of the two-point
functions become very large in modulus as $\D$ goes to 1.
The representation becomes a huge sum of large terms that
alternate in sign and sum up to a small number. This is
numerically delicate and requires to have a good accuracy for
the individual terms.

\begin{figure}
\centering
\includegraphics[width=.98\textwidth,clip=true]{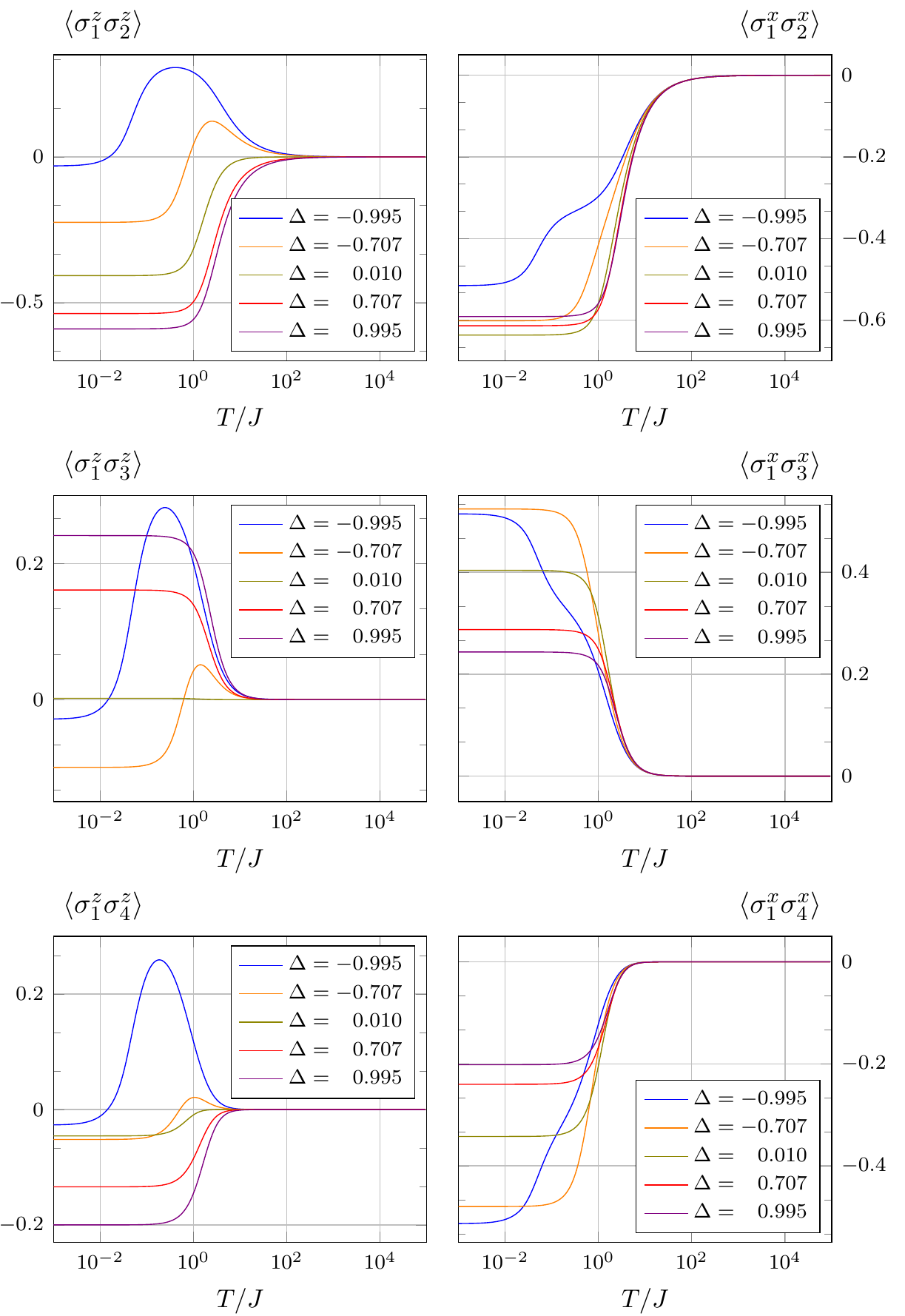}
\caption{\label{fig:n24_Tvar_Deltapar_hfix}
Short-range correlation functions ranging over up to four lattice
sites as a function of temperature for various values of $\D$ in the
massless regime and $h = 0$.}
\end{figure}
\begin{figure}
\centering
\includegraphics[width=.98\textwidth,clip=true]{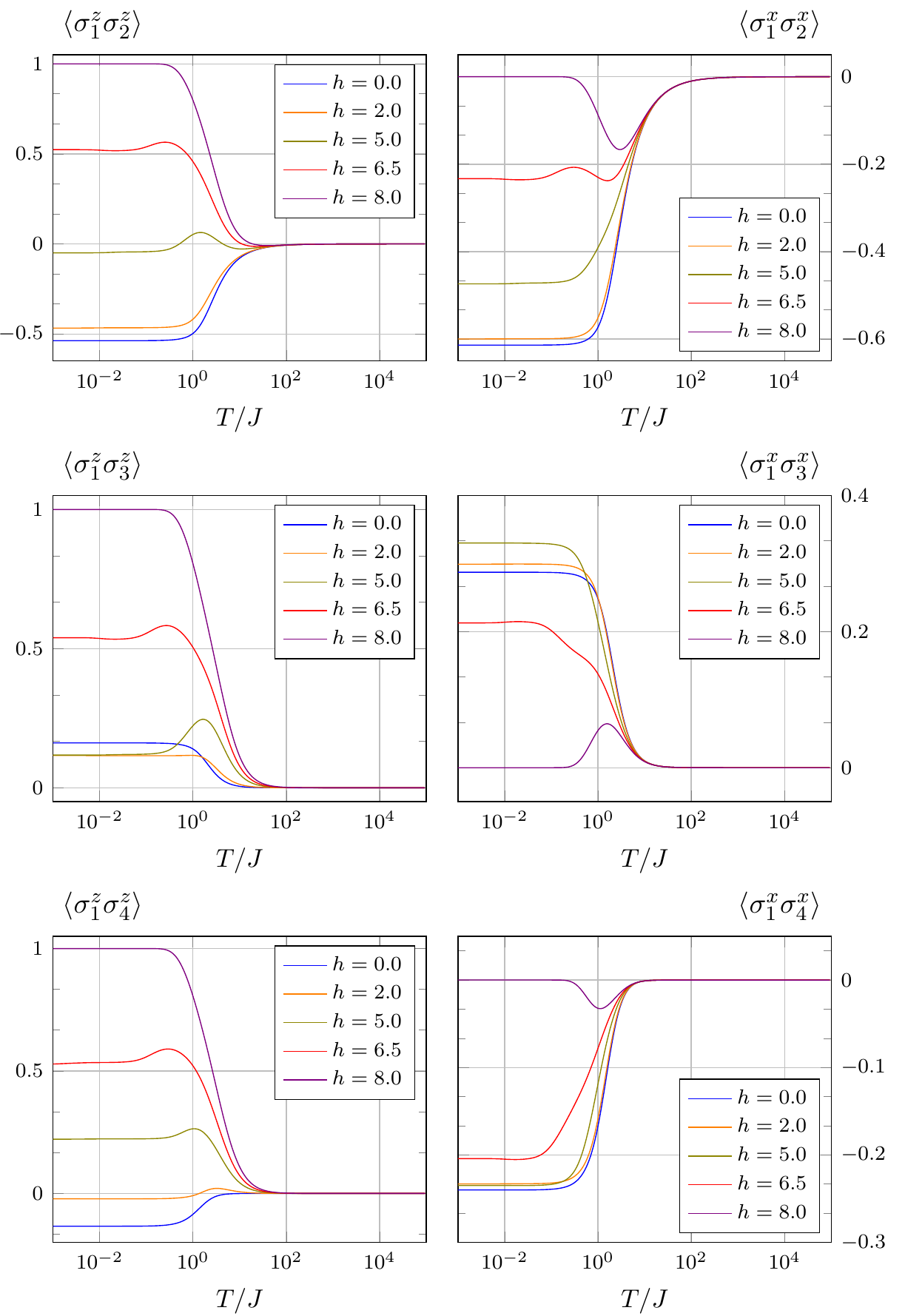}
\caption{\label{fig:n24_Tvar_hpar_Deltafix}
Short-range correlation functions ranging over up to four lattice
sites as a function of temperature for various values of the
magnetic field and $\D = 0.707$.}
\end{figure}
\begin{figure}
\centering
\includegraphics[width=.98\textwidth,clip=true]{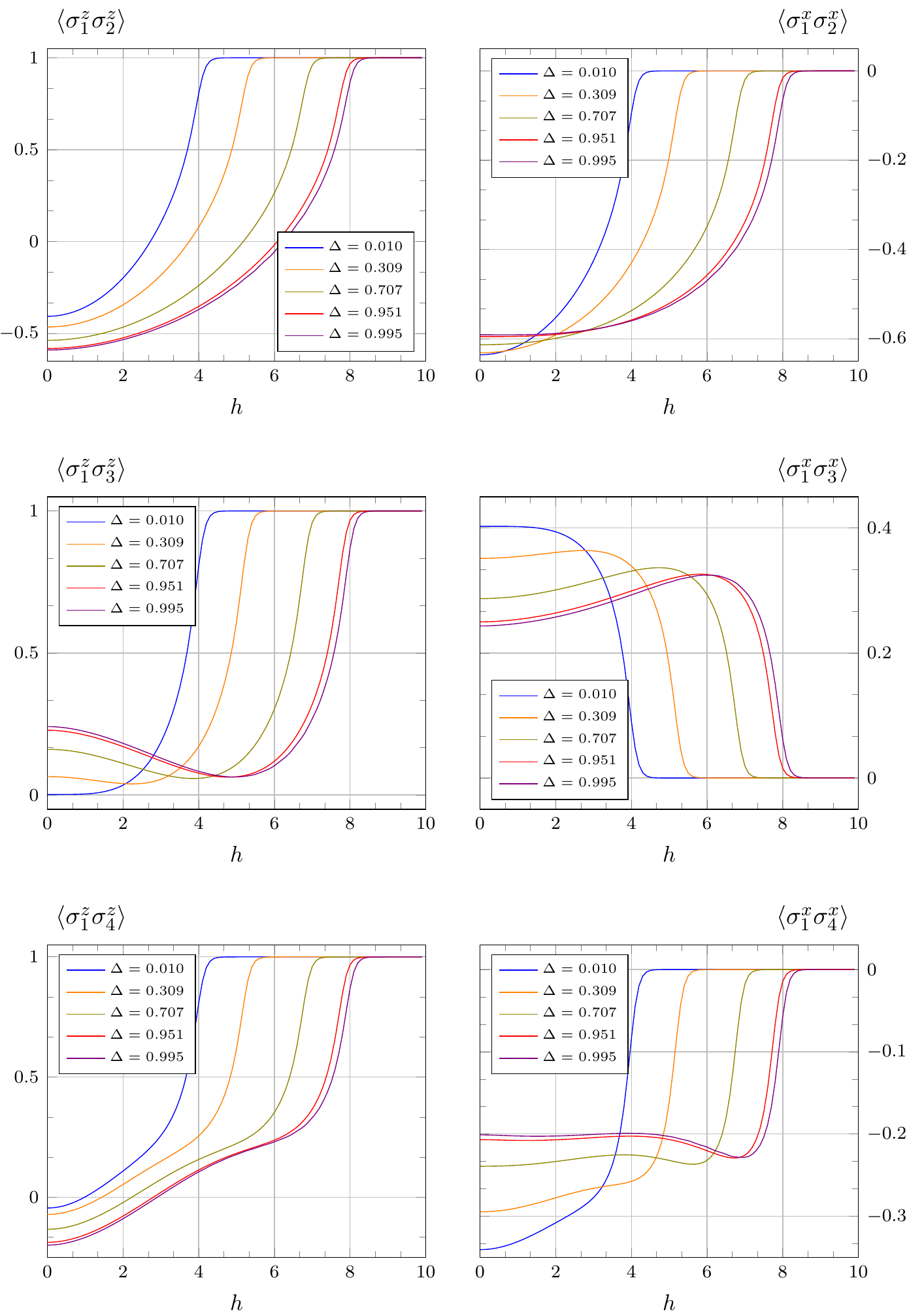}
\caption{\label{fig:n24_hvar_Deltapar_Tfix}
Short-range correlation functions ranging over up to four lattice
sites as a function of the magnetic field for various value of $\D$
and $T/J = 0.1$.}
\end{figure}

\subsection{Comparison with asymptotic results}
An interesting application of exact results is the test of
asymptotic formulae which often do not come with error estimates.
In the literature there are at least two different results
for the large-distance asymptotics of the static correlation
functions of the XXZ chain in the critical regime at zero and
small finite temperature, the work of Lukyanov and Terras
\cite{Lukyanov99,LuTe03}, in which fully explicit formulae
for $T = 0$ and $h = 0$ were derived, and the work \cite{DGK13a,%
DGK14a} treating the case of finite $h$ at small~$T$. The
question we would like to answer is, how large is large,
or, starting from which distance do the asymptotic formulae
provide reliable approximations to the correlation functions.

Lukyanov and Terras consider the long-distance asymptotic
behaviour of the two-point correlation functions
$\langle\sigma_1^x\sigma_n^x\rangle$ and
$\langle\sigma_1^z\sigma_n^z\rangle$ combining a Gaussian
conformal field theory with input from the $q$-vertex operator
approach applied to the XYZ chain \cite{Lashkevich02}.
They show that
\begin{multline} \label{lutexx}
  \langle\sigma_1^x\sigma_n^x\rangle
  \sim
  \frac{(-1)^nA}{n^\nu}\left\{ 1 - \frac{B}{n^{4/\nu-4}}
  + \mathcal{O}\left(n^{-2}\log n, n^{8-8/\nu}\right) \right\} \\[1ex]
  -\frac{\tilde A}{n^{\nu+1/\nu}}\left\{ 1 + \frac{\tilde B}{n^{2/\nu-2}}
  + \mathcal{O}\left(n^{-2}\log n, n^{4-4/\nu}\right) \right\} + \dots
\end{multline}
and
\begin{multline} \label{lutezz}
  \langle\sigma_1^z\sigma_n^z\rangle
  \sim
  -\frac{1}{\pi^2\nu n^2} \left\{ 1 + \frac{\tilde B_z}{n^{4/\nu-4}}
   \left(1 + \frac{2-\nu}{2(1-\nu)}\right)
   + \mathcal{O}\left(n^{-2}\log n, n^{8-8/\nu}\right) \right\} \\[1ex]
  +\frac{(-1)^n A_z}{n^{1/\nu}}\left\{ 1 - \frac{B_z}{n^{2/\nu-2}}
  + \mathcal{O}\left(n^{-2}\log n, n^{4-4/\nu}\right) \right\} + \dots
\end{multline}
where the functions $A,B,\tilde A, \tilde B, A_z,B_z,\tilde B_z$ depend
only on $\nu$ and are given explicitly in their work. Note that we
have adapted their formulae to our conventions by supplying a factor
of $(-1)^n$ to the transverse correlation functions. Lukyanov and Terras
are using the Hamiltonian $H_{\rm LT} = - H_L/2$ which is unitarily
equivalent to $H_L$ \cite{YaYa66a}. The unitary transformation is
induced by the adjoint action of $U = \prod_{j=-L}^L \s_{2j}^z$
accompanied by a reparametrization of $\D$ which we can take into
account with the identification
\begin{equation}
     \nu = 1 - \frac{\g}{\p} \epp
\end{equation}

For a comparison with our results we take the definition of the
function $\hat \om$ from \cite{BDGKSW08}. For $T=h=0$ the auxiliary
functions $\mathfrak{b}$ and $\bar{\mathfrak{b}}$ given in that
paper vanish. This reduces the calculation of the functions $\hat \om$
and $\hat \om'$ to the calculation of certain definite integrals
which is easily done numerically on a computer.

In (\ref{lutexx}), (\ref{lutezz}) it depends on $\nu$ which of the
terms is asymptotically dominant. For the sake of simplicity, we
separate the asymptotically dominant terms only if they can be identified
uniformly for all $0<\nu<1$. For this reason we consider two levels of
approximations for $\langle\sigma_1^x\sigma_n^x\rangle$: one consists
of only the first term $\frac{A}{n^\nu}$, which is the leading term
for general $\nu$, the other one consists of the whole expression. In
the case of $\langle\sigma_1^z\sigma_n^z\rangle$ we consider only the
whole expression. In this case the term containing $\tilde B_z$ can be
seen to be of higher order than the rest. Still, it makes no visible
difference whether we include it in our plots or not.

Looking at the figures \ref{fig:lukyanov_zz} and \ref{fig:lukyanov_xx}
showing $\langle\sigma_1^z\sigma_n^z\rangle$ and
$\langle\sigma_1^x\sigma_n^x\rangle$, respectively, we observe
that there are poles in the asymptotic expansion. For $\Delta\to-1$,
corresponding to $\nu\to0$, rapid oscillations are visible in all plots,
becoming less pronounced with increasing~$n$. These oscillations are not
a numerical error but rather a feature of the functions $B,\tilde{B},
B_z,\tilde{B}_z$. The poles visible in $\langle\sigma_1^x\sigma_n^x\rangle$
stem from the function $B$, which has poles of order 2 at
$\nu = \frac{2}{3+2l}$ for $l\in\mathbb{N}$. At these positions
$q$ is a root of unity. As can be seen in the plots, the poles
become narrower with increasing $n$.

\begin{figure}[htb]
\centering
\includegraphics{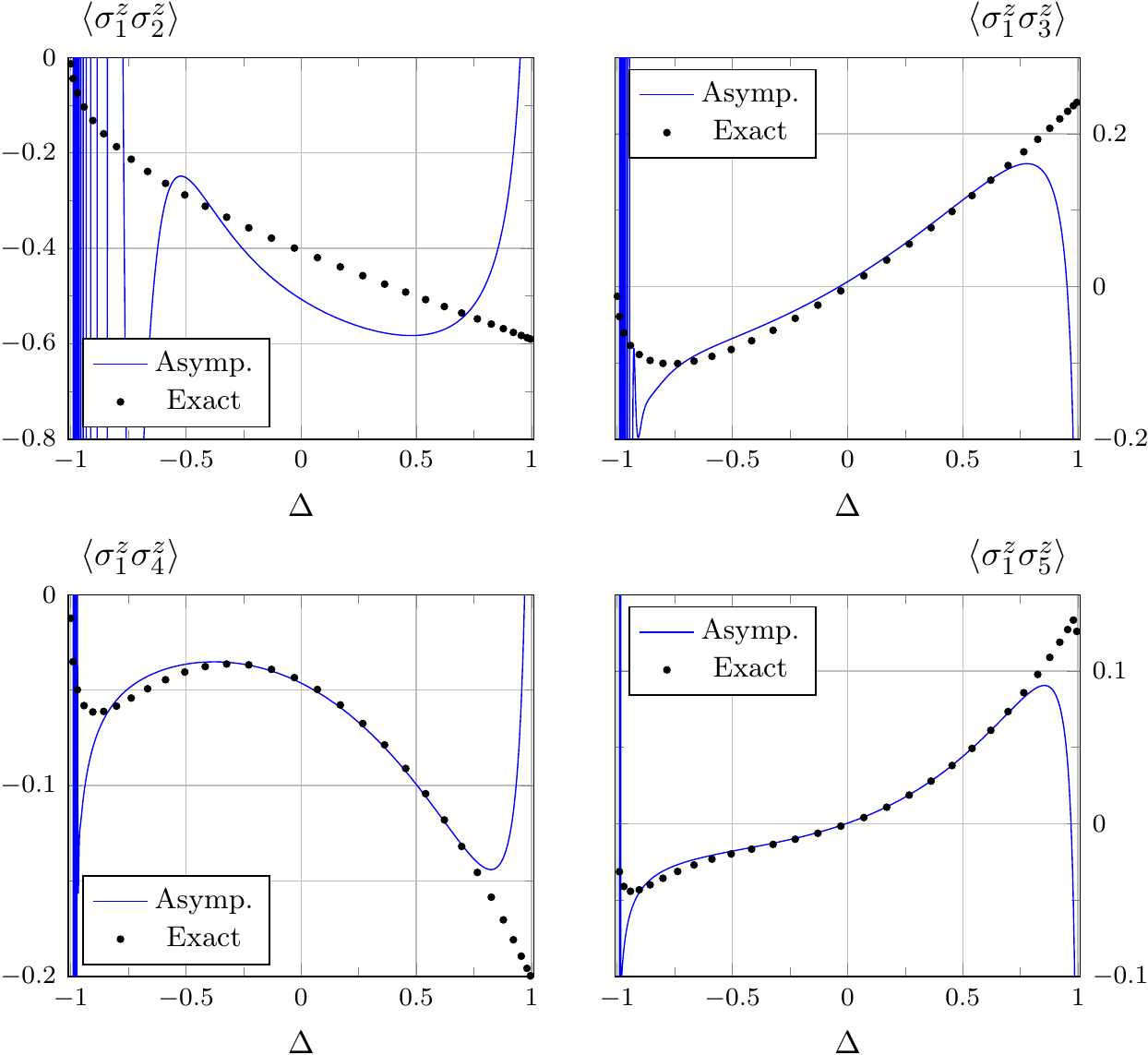}
\caption{Two-point correlators $\langle\sigma_1^z\sigma_n^z\rangle$
for $n=2,3,4,5$ in the ground state and at zero magnetic field in
the massless regime. Comparison of the asymptotic results of Lukyanov
and Terras with our exact results.}
\label{fig:lukyanov_zz}
\end{figure}
\begin{figure}[htb]
\centering
\includegraphics{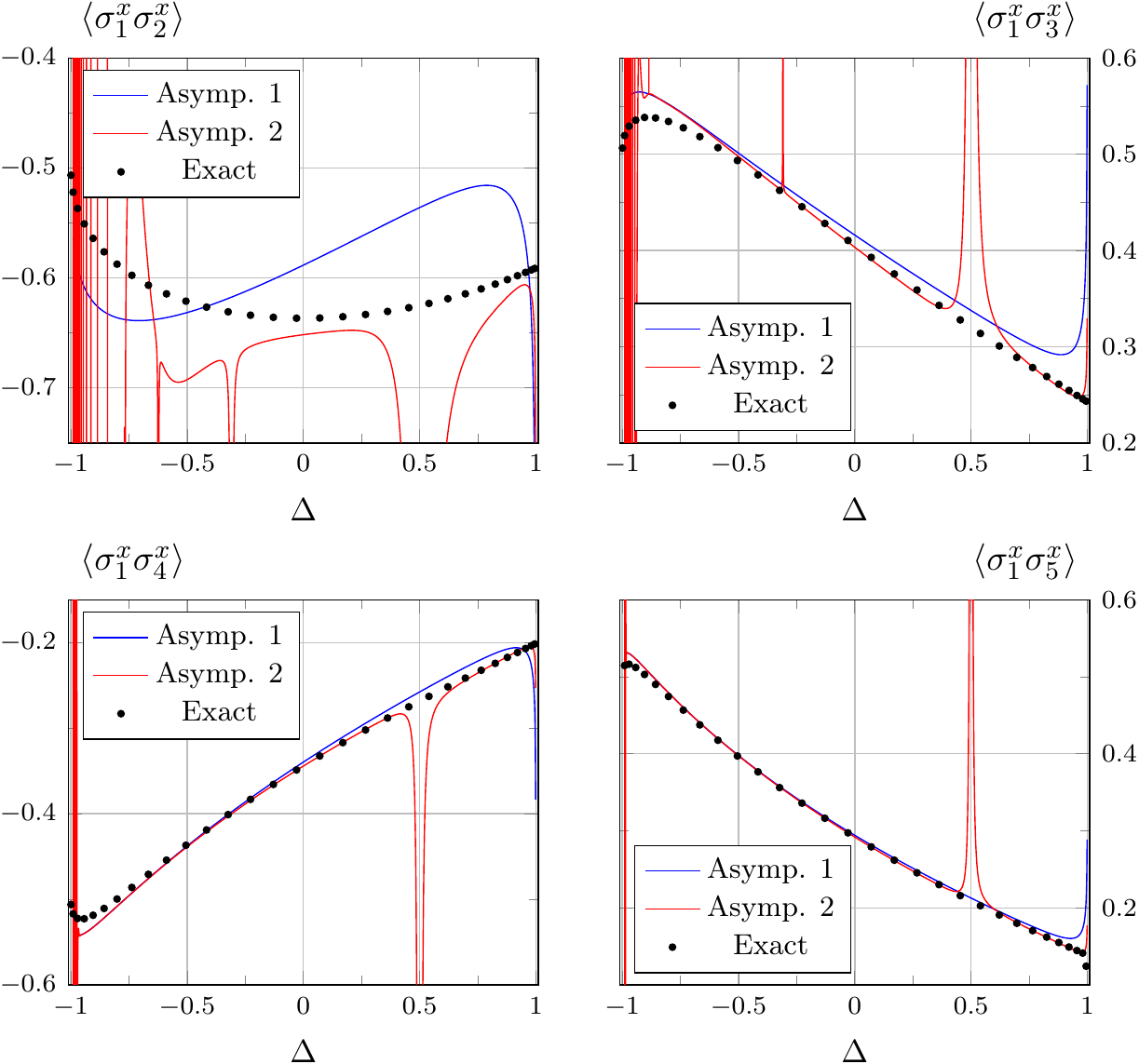}
\caption{Two-point correlators $\langle\sigma_1^x\sigma_n^x\rangle$
for $n=2,3,4,5$ in the ground state and at zero magnetic field in the
massless regime. Comparison of the asymptotic results of Lukyanov
and Terras with our exact results.} \label{fig:lukyanov_xx}
\end{figure}

Away from roots of unity, we observe the expected behaviour.
For the nearest-neighbour functions the asymptotic expansion
deviates considerably from the exact results. With increasing
$n$ the agreement between the results becomes better. It is
noteworthy, that the asymptotics agrees very well with the exact
values even for short distances, especially in the transverse
case.

In \cite{DGK13a,DGK14a} the two-point functions of the XXZ chain
in the low-temperature limit were studied within a thermal form
factor approach \cite{DGK13a}. In the low-$T$ limit the long-distance
asymptotics is determined by those terms in the form factor series
pertaining to the quantum transfer matrix of the model for which
the correlation lengths diverge for $T \rightarrow 0$. Using
a technique developed in \cite{KKMST11b} the authors of
\cite{DGK13a,DGK14a} were able to sum these contributions and
obtained, from a microscopic calculation, the expected asymptotic
behaviour of a conformal field theory on a cylinder. The formulae
for the amplitudes as functions of the magnetic field were new
and numerically efficient and complemented those for $h = 0$
of Lukyanov and Terras.

Let us briefly recall the main results of \cite{DGK13a,DGK14a}.
The required low-$T$ data are  the dressed charge function $Z$,
the density of Bethe roots $\rho$, and the dressed energy $\e$.
They are defined as the unique solutions of the Fredholm-type
integral equations
\begin{subequations}
\begin{align}
     Z(\la) & = 1 + \int_{-Q}^Q \frac{\rd \mu}{2\pi i} K_0 (\la-\mu) Z(\mu) \epc \\
     \rho(\la) & = - \frac{\re(\la+i\g/2)}{2\pi i}
        + \int_{-Q}^Q \frac{\rd \mu}{2\pi i} K_0 (\la-\mu) \rho(\mu) \epc \\
     \e(\la) & = \e_0(\la)
        + \int_{-Q}^Q \frac{\rd \mu}{2\pi i} K_0 (\la-\mu) \e(\mu) \epc \qd
          \e_0(\la) = h - \frac{4J(1-\D^2)}{\ch(2\la) - \D} \epp
\end{align}
\end{subequations}
The two points $\pm Q$ are called the Fermi points and $Q>0$ is
determined by $\e(Q) = 0$. In \cite{DGK14b} it was proven that
such a $Q$ exists and is unique. With these quantities we then
define the Fermi momentum $k_F$, the Fermi sound velocity $v_0$
and the dressed charge $\mathcal{Z}$ at the Fermi point,
\begin{equation}
     k_F = 2\pi \int_0^Q \rd\la\ \rho(\la) \epc \qd
     v_0 = \frac{\e'(Q)}{2\pi\rho(Q)} \epc \qd
     \mathcal{Z} = Z(Q) \epp
\end{equation}

The asymptotic expressions consist of products of amplitudes times
terms that oscillate and decay with distance. The amplitudes
$A_{0,n}^{zz}$ and $A_{0,0}^{-+}$ are slightly complicated expressions
given in equations (90) and (97b) of \cite{DGK14a}. We refrain from
reproducing them here. The leading oscillating and decaying contribution
can be expressed in terms of the above defined functions. For the
longitudinal case the asymptotic behaviour takes the form
\begin{multline}
     \langle\s_1^z\s_{m+1}^z\rangle
        - \langle\s_1^z\rangle\langle\s_{m+1}^z\rangle \\
        \sim A_{0,0}^{zz}\lb(\frac{\pi T/v_0}{\sh(m\pi T/v_0)}\rb)^2
	+ A_{0,1}^{zz} \cos(2m k_F)
	  \lb(\frac{\pi T/v_0}{\sh(m\pi T/v_0)}\rb)^{2\mathcal{Z}^2} \epp
\end{multline}
Here the term containing $A_{0,0}^{zz}$ is the leading term for $\D < 0$,
whereas the term with $A_{0,1}^{zz}$ is dominant for $\D > 0$. In the
transverse case the asymptotic behaviour is described by
\begin{equation}
  \langle\s_1^-\s_{m+1}^+\rangle \sim
     A_{0,0}^{-+} (-1)^m
     \lb(\frac{\pi T/v_0}{\sh(m\pi T/v_0)}\rb)^{\frac{1}{2\mathcal{Z}^2}} \epp
\end{equation}
It should be noted that these expressions are numerically efficient
and can be evaluated on a laptop computer in rather short time. Most
of the numerical cost goes into the calculation of the amplitudes which
are independent of the distance and of the temperature and have
to be computed only once for given values of the anisotropy parameter
and of the magnetic field.

\begin{figure}
\centering
\includegraphics{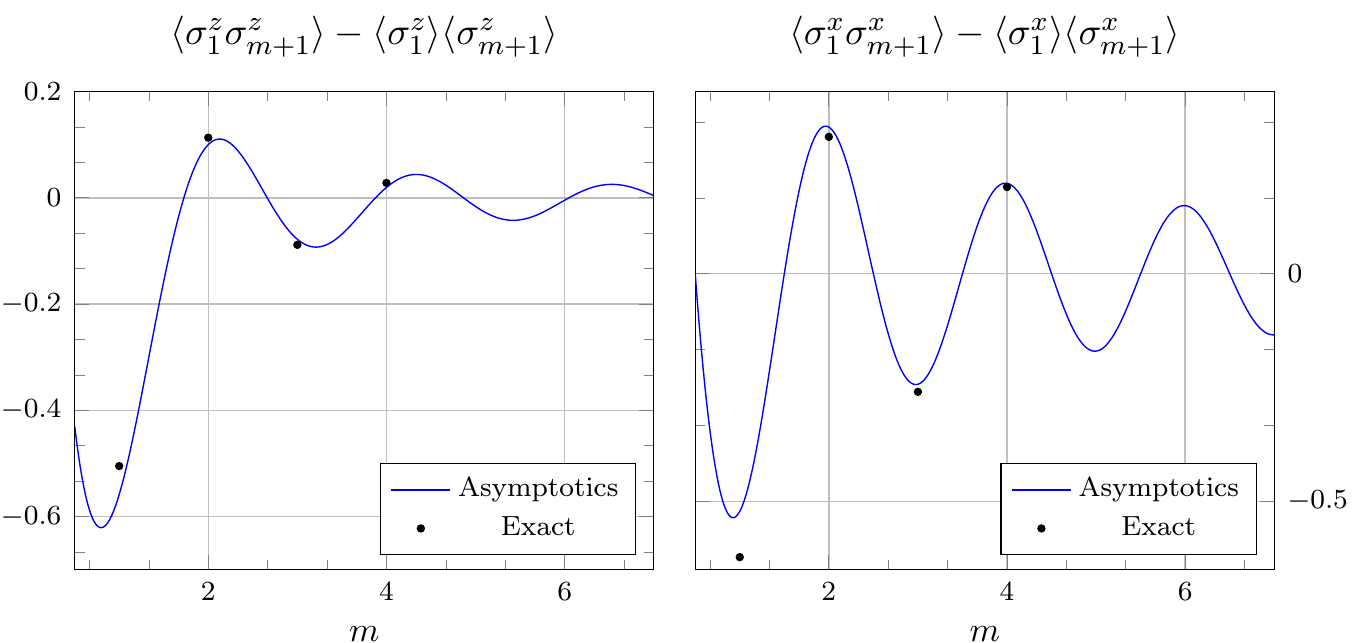}
\caption{Comparison between the asymptotic expansion obtained in
\cite{DGK13a,DGK14a} and our exact results for $\D=0.6$, $h=1$
and $T/J=0.04$.}
\label{fig:formfactor_vs_distance}
\end{figure}

Fig.~\ref{fig:formfactor_vs_distance} shows the comparison between
the asymptotic and exact results as functions of the distance $m$.
It can be seen that the asymptotics come very close to the exact
results for surprisingly small distances, starting with $m=3,4$.
This is of course dependent on the chosen parameters. For example,
close to the isotropic point the agreement becomes worse.
Fig.~\ref{fig:formfactor_vs_h_T01} shows both results as a
function of the external field $h$. Again, for distance $m=4$
the agreement is remarkable even for a non-trivial structure as
shown for the longitudinal case.

\begin{figure}[htb]
\centering
\includegraphics{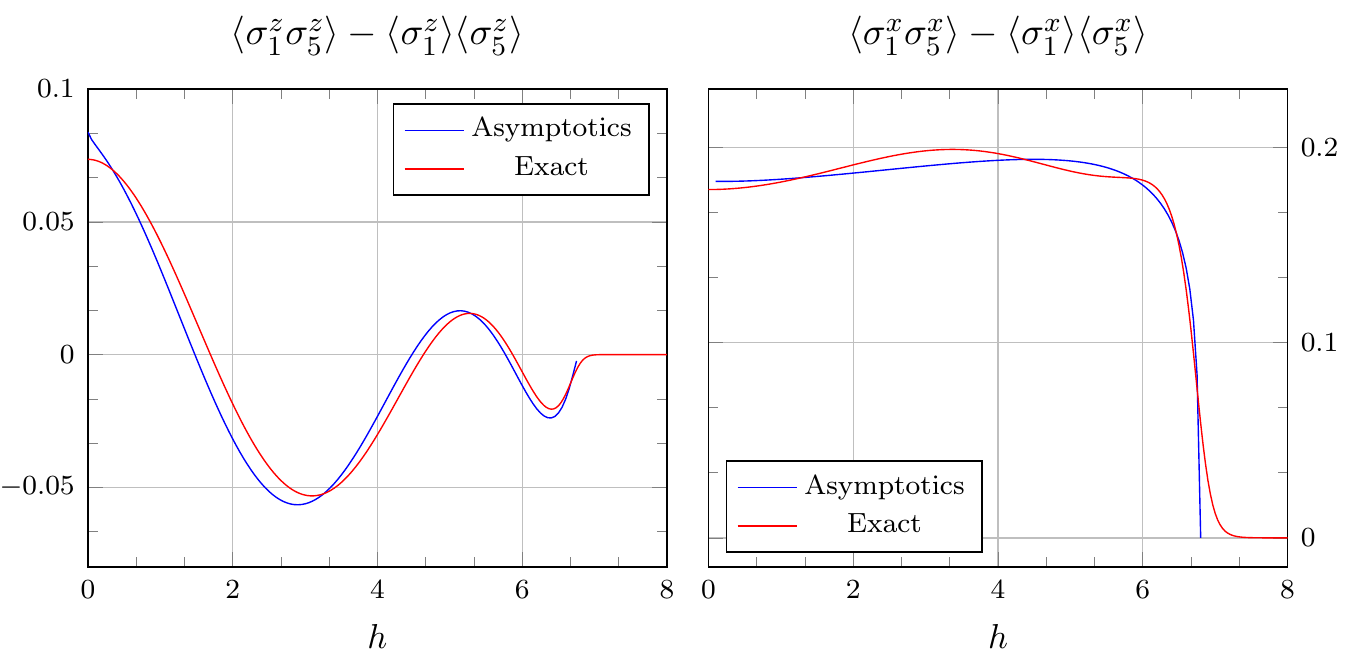}
\caption{Comparison of exact and asymptotic behaviour as a function
of the magnetic field for $n = 5$, $\D=0.7$ and $T/J=0.1$.}
\label{fig:formfactor_vs_h_T01}
\end{figure}

We can conclude that the asymptotic formulae derived in
\cite{DGK13a,DGK14a} are very close to the exact results for
surprisingly small values of the distance $m$. A rough
estimation for the temperatures for which the asymptotics
are valid would be $T/J < 0.1$. In addition, this comparison
provides another test to see that our results are consistent
with previous works.

\section{Conclusions}
\label{sec:concl}
We have given a descriptive review of the Fermionic basis
approach to the theory of correlation functions of the XXZ chain.
In the course of the review we worked out a few details that
were omitted in the original literature. Our main interest in
this work was to explore the efficiency of the exponential form
in the homogenous case for the actual computation of short-range
correlation functions. For this purpose we wrote out the
mode expansions in explicit form and worked out the explicit
formulae for the action of the modes on operators of finite
length.

We further worked out explicitly the expressions for $\<\s_1^x\s_n^x\>$
and $\<\s_1^z\s_n^z\>$ for $n = 2, 3, 4, 5$ in terms of
$\hat \om$, $\hat \om'$ and the derivatives of these functions.
The lengthy result for $n = 5$ was heretofore unknown and
is listed in Appendix~\ref{app:n5_explicit}. We believe that
the corresponding formulae for $n$ up to 9 or 10 could be worked
out in a similar way. However, the length of the final answer
will be rapidly growing. It would fill many pages, and printing
it out would make little sense. What would rather be needed
would be a better understanding of the structure of such
formulae. It was part of the motivation of the work \cite{FrSm18}
to guess this structure, but so far this attempt was only
partially successful.

Using the explicit Fermionic basis expansions we have discussed
in some details the two-point functions for the system
coupled to a heat bath of temperature $T$. For this case
we have also compared the exact correlation functions in the
critical regime with asymptotic formulae for their large-distance
behaviour. These asymptotic formulae, valid in the low-$T$
limit, turned out to be very good approximations even for distances
as short as three or four lattice sites.

\noindent
{\bf Acknowledgments.} We would like to thank Herman Boos for
numerous helpful discussions and Constantin Babenko and Herman
Boos for a careful reading of the manuscript. FG and RK acknowledge
financial support by the DFG in the framework of the research
unit FOR 2316.

\renewcommand{\thesection}{\Alph{section}}
\renewcommand{\theequation}{\thesection.\arabic{equation}}
\begin{appendices}

\section{From the JMS theorem to the exponential form}
\setcounter{equation}{0}
\label{app:jms_to_expform}
The JMS theorem \cite{JMS08} is central to the Fermionic basis
approach. Here we show how the exponential form (\ref{expkappa})
of the reduced density matrix that had appeared earlier in the
literature \cite{BGKS07,BJMST05b,BJMST08a} is a natural consequence
of the JMS theorem.

We define
\begin{equation} \label{defomhut}
     \hat \Om = \int_\G \frac{\rd \z_1^2}{2 \p \i \z_1^2}
           \int_\G \frac{\rd \z_2^2}{2 \p \i \z_2^2}
	   \om(\z_1, \z_2; \a) \bv (\z_1) \cv (\z_2) \epc
\end{equation}
and
\begin{equation}
     B (\z) = \int_\G \frac{\rd \x^2}{2 \p \i \x^2}
                \om(\x, \z; \a) \bv (\x) \epc \qd
     C (\z) = \int_\G \frac{\rd \x^2}{2 \p \i \x^2}
                \om(\z, \x; \a) \cv (\x) \epc
\end{equation}
where $\G$ is small circle around $1$.
\begin{lemma}
\label{lem:expformbcstar}
\begin{equation}
     [\re^{\hat \Om}, \bv^*(\z)] = - C(\z) \re^{\hat \Om} \epc \qd
     [\re^{\hat \Om}, \cv^*(\z)] = B(\z) \re^{\hat \Om} \epp
\end{equation}
\end{lemma}
\begin{proof}
Using (\ref{anticommutism}) we see that $[\hat \Om, \bv^*(\z)]
= - C(\z)$, implying $[\hat \Om^k, \bv^*(\z)] = - C(\z) k \hat \Om^{k-1}$.
Hence,
\begin{equation}
     [\re^{\hat \Om}, \bv^*(\z)] = [\re^{\hat \Om} - 1, \bv^*(\z)]
        = - C(\z) \sum_{k=1}^\infty \frac{\hat \Om^{k-1}}{(k-1)!}
	= - C(\z) \re^{\hat \Om}
\end{equation}
which is the first identity. The second one follows in a similar
way.
\end{proof}
\begin{lemma}
\label{lem:zkappavac}
\begin{subequations}
\label{zkappavac}
\begin{align}
     Z^\k \bigl\{\re^{\hat \Om} \tv^* (\z) X^{(\a)}\bigr\} & =
        2 \r(\z) Z^\k \bigl\{\re^{\hat \Om} X^{(\a)}\bigr\} \epc \\[1ex]
     Z^\k \bigl\{\re^{\hat \Om} \bv^* (\z) X^{(\a+1)}\bigr\} & = 0 \epc
        \label{kappaleftvacbs} \\[1ex]
     Z^\k \bigl\{\re^{\hat \Om} \cv^* (\z) X^{(\a-1)}\bigr\} & = 0 \epp
        \label{kappaleftvaccs}
\end{align}
\end{subequations}
\end{lemma}
\begin{proof}
The first equation follows from the JMS theorem, since
$[\re^{\hat \Om}, \tv^* (\z)] = 0$. The JMS also implies
that $Z^\k \bigl\{(\bv^*(\z) - C(\z)) \re^{\hat \Om} 
X^{(\a + 1)}\bigr\} = Z^\k \bigl\{(\cv^*(\z) +
B(\z)) \re^{\hat \Om} X^{(\a - 1)}\bigr\} = 0$. Combining
these two identities with Lemma~\ref{lem:expformbcstar}
we obtain (\ref{kappaleftvacbs}), (\ref{kappaleftvaccs}).
\end{proof}
Lemma~\ref{lem:zkappavac} means that the functional
$Z^\k\bigl\{\re^{\hat \Om} \cdot\bigr\}$ acts as a left vacuum
for the creation operators $\bv^* (\z)$, $\cv^* (\z)$. In general,
this functional is hard to evaluate. A very special realization
is obtained if we use the finite Trotter number approximant
to the reduced density matrix $D^N_{[k,l]} (T,\k,\a)$ of the
canonical ensemble for the definition of $Z^\k$. This means
to put spin-$\2$ representations on the horizontal lines in
(\ref{gendensmat}) and, in an alternating manner, spectral
parameters $\pm \frac{c}{NT}$, where $c$ is an appropriate
constant (for more details see e.g.\ \cite{Goehmann20}).
Sending then $T \rightarrow \infty$ at fixed $\k$ all
$R$-matrices in (\ref{gendensmat}) degenerate into permutation
matrices and the left and right eigenvectors become independent
of $\a$ and $\k$ (for a graphical representation of the limit
see Fig.~\ref{fig:hightgraph}). All in all we see that
$Z^\k$ has the limit
\begin{equation}
     \lim_{T \rightarrow \infty} Z^\k \bigl\{X^{(\a)}\bigr\}
        = \frac{\tr_{[k,l]} \bigl\{q^{2 \k S_{[k,l]}} X_{[k,l]}\bigr\}}
	       {\tr_{[k,l]} \bigl\{q^{2 \k S_{[k,l]}}\bigr\}}
        = \tr^{-2\k} \bigl\{X^{(\a)}\bigr\} \epc
\end{equation}
where we used the definition (\ref{defkappatrace}) of the $\k$-trace.
\begin{figure}
\centering
\includegraphics[width=.70\textwidth,angle=0,clip=true]{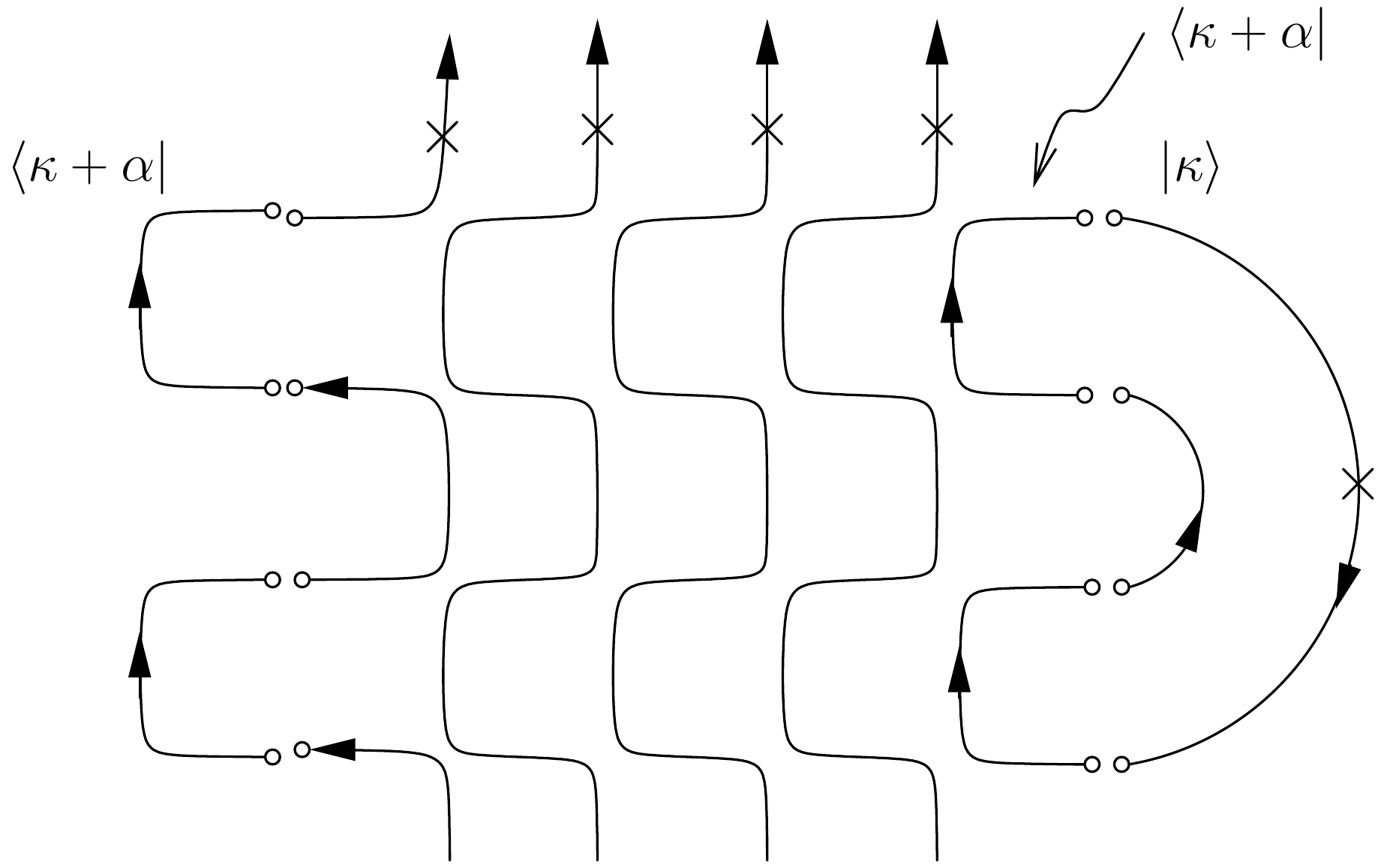}
\caption{\label{fig:hightgraph} 
Graphical representation of the (unnormalized) reduced density
matrix $D^N_{[k,l]} (T,\k,\a)$ in the limit $T \rightarrow \infty$.
Crosses denote operators $q^{\k \s^z}$.} 
\end{figure}

The limits $T \rightarrow \infty$ and $N \rightarrow \infty$
commute for $D^N_{[k,l]} (T,\k,\a)$ \cite{GGKS20}. The unique
solution of (\ref{nlietemhom}) for $T \rightarrow \infty$ is
$\fa_0 (\la, \k) = q^{-2\k}$. Inserting it into (\ref{rhoint})
we obtain
\begin{equation}
     \lim_{T \rightarrow \infty} \r (\z)
        = \frac{q^{\k + \a} + q^{- \k - \a}}{q^\k + q^{- \k}} \epp
\end{equation}
Then, using the latter in (\ref{newg})-(\ref{ompsi}), we see
that the high-$T$ limit of $\om$ is
\begin{equation}
     \lim_{T \rightarrow \infty} \om (\z_1, \z_2; \a)
        = - \biggl(\frac{q^\k - q^{- \k}}{q^\k + q^{- \k}}\biggr)^2
	    \D \ps (\z_1/\z_2, \a) \epp
\end{equation}
For the special value $\k = - \a/2$ we have
\begin{equation}
     \lim_{T \rightarrow \infty} \r (\z)\bigr|_{\k = - \a/2} = 1 \epc \qd
     \lim_{T \rightarrow \infty} \om (\z_1, \z_2; \a)\bigr|_{\k = - \a/2}
        = \om_0 (\z_1/\z_2; \a) \epc
\end{equation}
where $\om_0$ was defined in (\ref{defomzero}). Setting
\begin{equation} \label{defomnullhut}
     \hat \Om_0 = \int_\G \frac{\rd \z_1^2}{2 \p \i \z_1^2}
           \int_\G \frac{\rd \z_2^2}{2 \p \i \z_2^2}
	   \om_0 (\z_1/\z_2; \a) \bv (\z_1) \cv (\z_2)
\end{equation}
and taking the limit $T \rightarrow \infty$ at $\k = - \a/2$ in
(\ref{zkappavac}) we conclude that the following lemma holds true.
\begin{lemma}
\label{lem:tralphavac}
\begin{subequations}
\label{tralphavac}
\begin{align}
     \tr^\a \bigl\{\re^{\hat \Om_0} \tv^* (\z) X^{(\a)}\bigr\} & =
        2 \tr^\a \bigl\{\re^{\hat \Om_0} X^{(\a)}\bigr\} \epc \\[1ex]
     \tr^\a \bigl\{\re^{\hat \Om_0} \bv^* (\z) X^{(\a+1)}\bigr\} & = 0 \epc
        \label{traleftvacbs} \\[1ex]
     \tr^\a \bigl\{\re^{\hat \Om_0} \cv^* (\z) X^{(\a-1)}\bigr\} & = 0 \epp
        \label{traleftvaccs}
\end{align}
\end{subequations}
\end{lemma}
This lemma means that the functional $\tr^\a \bigl\{\re^{\hat \Om_0}
\cdot \bigr\}$ can be interpreted as a left vacuum for the 
creation operators $\bv^* (\z)$ and $\cv^* (\z)$. The lemma
can be used to prove the validity of the exponential form,
equation (\ref{expkappanull}) of the main text, for $\k = 0$.
For this purpose we introduce some additional notation. Let
\begin{equation}
     {\cal B}^{(l)} = \bv^* (\z_1^+) \dots \bv^* (\z_l^+) \epc \qd
     {\cal C}^{(l)} = \cv^* (\z_l^-) \dots \cv^* (\z_1^-) \epp
\end{equation}
By ${\cal B}^{(l)}_j$ we denote ${\cal B}^{(l)}$ with
the $j$th factor from the left omitted, whereas ${\cal C}^{(l)}_k$
will stand for ${\cal C}^{(l)}$ with the $k$th factor from the
right omitted. We would further like to recall the definition
(\ref{defomop}) of the operator $\Om$. With (\ref{defomhut}) and
(\ref{defomnullhut}) it can be expressed as $\Om = \hat \Om_0
- \hat \Om$.
\begin{lemma}
\label{lem:laplacerecursion}
\begin{equation}
     \tr^\a \bigl\{\re^{\Om} {\cal B}^{(l)} {\cal C}^{(l)} q^{2 \a S(0)}\bigr\}
        = \sum_{k=1}^l (-1)^{j+k} \om(\z_j^+, \z_k^-; \a)
	     \tr^\a \bigl\{\re^{\Om} {\cal B}^{(l)}_j
	                   {\cal C}^{(l)}_k q^{2 \a S(0)}\bigr\} \epp
\end{equation}
\end{lemma}
\begin{proof}
\begin{align}
     \tr^\a \bigl\{\re^{\Om} {\cal B}^{(l)} {\cal C}^{(l)}& q^{2 \a S(0)}\bigr\}
        \notag \\[1ex] & =
     (-1)^{j-1} \tr^\a \bigl\{\re^{\hat \Om_0} \re^{- \hat \Om} \bv^* (\z_j^+)
                              {\cal B}^{(l)}_j {\cal C}^{(l)} q^{2 \a S(0)}\bigr\}
        \notag \\[1ex] & =
     (-1)^{j-1} \tr^\a \bigl\{\re^{\hat \Om_0}
                       \bigl(\bv^* (\z_j^+) + C(\z_j^+)\bigr)\re^{- \hat \Om}
                              {\cal B}^{(l)}_j {\cal C}^{(l)} q^{2 \a S(0)}\bigr\}
        \notag \\[1ex] & =
     (-1)^{j+l} \tr^\a \bigl\{\re^{\Om} {\cal B}^{(l)}_j C(\z_j^+)
		              {\cal C}^{(l)} q^{2 \a S(0)}\bigr\}
        \notag \\ & =
     \sum_{k=1}^l (-1)^{j+k}
        \tr^\a \bigl\{\re^{\Om} {\cal B}^{(l)}_j [C(\z_j^+), \cv^* (\z_k^-)]_+
	              {\cal C}^{(l)}_k q^{2 \a S(0)}\bigr\}
        \notag \\ & =
     \sum_{k=1}^l (-1)^{j+k} \om(\z_j^+, \z_k^-; \a)
	     \tr^\a \bigl\{\re^{\Om} {\cal B}^{(l)}_j
	                   {\cal C}^{(l)}_k q^{2 \a S(0)}\bigr\} \epp
\end{align}
Here we used Lemma~\ref{lem:expformbcstar} in the second equation and
Lemma~\ref{lem:tralphavac} in the third equation.
\end{proof}
Let
\begin{equation}
     {\cal T}^{(k)} = \tv^* (\z_1^0) \dots \tv^* (\z_k^0) \epp
\end{equation}
\begin{lemma}
\label{lem:expformonbasis}
\begin{equation}
     \tr^\a \bigl\{\re^{\Om} {\cal T}^{(k)} {\cal B}^{(l)}
                             {\cal C}^{(l)} q^{2 \a S(0)}\bigr\}
        = 2^k \det_{m, n = 1, \dots, l} \bigl\{\om(\z_m^+, \z_n^-; \a)\bigr\} \epp
\end{equation}
\end{lemma}
\begin{proof}
The operators ${\cal T}^{(k)}$ and $\re^{\hat \Om}$ commute.
It follows from Lemma~\ref{lem:tralphavac} that
\begin{equation}
     \tr^\a \bigl\{\re^{\Om} {\cal T}^{(k)} {\cal B}^{(l)}
                             {\cal C}^{(l)} q^{2 \a S(0)}\bigr\}
        = 2^k \tr^\a \bigl\{\re^{\Om} {\cal B}^{(l)} {\cal C}^{(l)}
	                    q^{2 \a S(0)}\bigr\} \epp
\end{equation}
Furthermore, since $\re^\Om q^{2 \a S(0)} = \id$,
Lemma~\ref{lem:laplacerecursion} in conjunction with the Laplace
expansion formula for determinants implies that
\begin{equation}
     \tr^\a \bigl\{\re^{\Om} {\cal B}^{(l)} {\cal C}^{(l)}
                   q^{2 \a S(0)}\bigr\} =
        \det_{m, n = 1, \dots, l} \bigl\{\om(\z_m^+, \z_n^-; \a)\bigr\} \epp
\end{equation}
\end{proof}

Now recall that the function $\r$ was originally defined as
an eigenvalue ratio in (\ref{defrho}). If the eigenvalue
in the definition of $\r$ is non-degenerate, $\r$ is an
even function of its second argument $\k$. This is, for instance,
the case with the dominant eigenvalue of the quantum transfer
matrix that occurs in the description of the finite temperature
reduced density matrix. Then (\ref{defrho}) implies that
$\r \bigr|_{\k = - \a/2} = 1$. Combining Corollary~\ref{cor:zkappagenfns}
and Lemma~\ref{lem:expformonbasis} we conclude that 
\begin{equation}
     Z^{-\a/2} \bigl\{{\cal T}^{(k)} {\cal B}^{(l)}
                      {\cal C}^{(l)} q^{2 \a S(0)}\bigr\} =
     \tr^\a \bigl\{\re^{\Om} {\cal T}^{(k)} {\cal B}^{(l)}
                   {\cal C}^{(l)} q^{2 \a S(0)}\bigr\} \epp
\end{equation}
Inserting the mode expansions and using Theorem~\ref{thm:basis}
we arrive at equation (\ref{relztralpha}) and (\ref{expkappanull})
of the main text.

We claim that (\ref{expkappanull}) remains valid for non-zero
values of $\k$ if we restrict the action of $Z^\k$ to spin-reversal
invariant operators such as $\s_1^z\s_n^z$ or $\s_1^x\s_n^x$.
We have checked this by direct use of the Fermionic basis for
$n = 1, 2, 3$ \cite{Kleinemuehl20}. It also follows if we assume
the existence of an operator $\tv (\z)$ conjugate to $\tv^* (\z)$
\cite{Kleinemuehl20}. Such an operator has been defined by its
properties in \cite{BoGo09} in the inhomogeneous case and in
\cite{Kleinemuehl20} in the homogeneous case. In the inhomogeneous
case we have checked that the postulated properties are
sufficient to fix the operator $\tv (\z)$ for $n = 1, 2$ in
\cite{BoGo09} and for $n = 3$ in \cite{Weisse09}.

\section{Multiplicative and additive spectral parameters}
\label{app:multtoadd}
\setcounter{equation}{0}
The following formula can be used to switch from the definition
of the function $\om$ with multiplicative spectral parameters
to the corresponding function $\hat \om$ with additive spectral
parameters favoured in \cite{BDGKSW08}. Let $u = \ln (\z)$. Then
\begin{multline}
     \6_{\z^2}^k \z^{2 \ell} f \bigl(\ln (\z)\bigr) =
	\biggl( \2 \re^{-2u} \6_u \biggr)^k \re^{2 u \ell} f(u) = \\
	\re^{2u (\ell -k)} (-1)^k \bigl( - \ell - \tst{\2} \6_u \bigr) \dots
	   \bigl( - \ell - \tst{\2} \6_u + k - 1 \bigr) f(u)
\end{multline}
for all $\ell \in {\mathbb Z}$. Introducing
the Pochhammer symbol
\begin{equation}
     (x)_k = x (x + 1) \dots (x + k - 1)
\end{equation}
we obtain
\begin{equation}
     \6_{\z^2}^k \, \z^{2 \ell} f \bigl(\ln (\z)\bigr)\bigr|_{\z^2 = 1} =
        (-1)^k \bigl( - \ell - \tst{\2} \6_u \bigr)_k f(u) \bigr|_{u = 0} \epc
\end{equation}
which can be nicely implemented on a computer.

\section{\boldmath Two-point functions for $n = 5$}
\setcounter{equation}{0}
\label{app:n5_explicit}
This appendix contains the explicit expressions for the two
independent two-point functions $\<\s_1^x \s_5^x\>$ and
$\<\s_1^z\s_5^z\>$.
\newcommand{\omegahat}{\hat \omega}%
In the following we shall employ the shorthand notation
$\omegahat_{ij} = \partial_\la^i\partial_\mu^j\omegahat(\la,\mu)_{\la=\mu=0}$ and
$\omegahat_{ij}' = \partial_\la^i\partial_\mu^j\omegahat'(\la,\mu)_{\la=\mu=0}$.

\end{appendices}


\providecommand{\bysame}{\leavevmode\hbox to3em{\hrulefill}\thinspace}
\providecommand{\MR}{\relax\ifhmode\unskip\space\fi MR }
\providecommand{\MRhref}[2]{%
  \href{http://www.ams.org/mathscinet-getitem?mr=#1}{#2}
}
\providecommand{\href}[2]{#2}

\end{document}